\documentclass[journal]{IEEEtran}

\usepackage{cite}
\usepackage{amsmath,amssymb,amsfonts,mathrsfs,physics,bm,amsthm}
\usepackage{algorithmic} 
\usepackage{graphicx}
\usepackage{textcomp}
\usepackage{xcolor}
\usepackage[utf8]{inputenc}
\def\BibTeX{{\rm B\kern-.05em{\sc i\kern-.025em b}\kern-.08em
    T\kern-.1667em\lower.7ex\hbox{E}\kern-.125emX}}
\makeatletter
   \def\hrulefill{\leavevmode\leaders\hrule height 1pt\hfill\kern\z@}
\makeatother
\usepackage{wrapfig}
\usepackage{caption}
\usepackage{subcaption}
\usepackage{float}

\newtheorem*{proposition1}{Proposition~I}

\newtheorem*{corollary}{Corollary I}

\begin{document}

\title{Doppler Estimation for High--Velocity  Targets \\ Using Subpulse Processing and the Classic \\  Chinese Remainder Theorem}

\author{Fernando Darío Almeida García, André Saito Guerreiro, Gustavo Rodrigues de Lima Tejerina, \\ José Cândido S. Santos Filho, Gustavo~Fraidenraich, and Michel Daoud Yacoub,~\IEEEmembership{Member,~IEEE} 
\thanks{F. D. A. García, A. S. Guerreiro, G. R. de Lima Tejerina, J. C. S. Santos Filho, G. Fraidenraich, and M. D. Yacoub are with the Wireless Technology Laboratory, Department of Communications, School of Electrical and Computer Engineering, University of Campinas, 13083-852 Campinas, SP, Brazil, Tel:+55(19)3788-5106, e-mail:$\left\{\right.$ferdaral, andsaito, tejerina, candido, gf, michel$\left.\right\}$@decom.fee.unicamp.br.}  
}

\markboth{}%
{Almeida \MakeLowercase{\textit{et al.}}: }

\maketitle

\begin{abstract}
In pulsed Doppler radars, the classic Chinese remainder theorem (CCRT) is a common method to resolve Doppler ambiguities caused by fast-moving targets.
Another issue concerning high-velocity targets is related to the loss in the signal-to-noise ratio (SNR) after performing \textit{range compression}. 
In particular, this loss can be partially mitigated by the use of subpulse processing (SP). 
Modern radars combine these techniques in order to reliably unfold the target velocity.
However, the presence of background noise may compromise the Doppler estimates. Hence, a rigorous statistical analysis is imperative.
In this work, we provide a comprehensive analysis on Doppler estimation. In particular, we derive novel closed-form expressions for the probability of detection (PD) and probability of false alarm (PFA).
To this end, we consider the newly introduce SP along with the CCRT.
A comparison analysis between SP and the classic pulse processing (PP) technique is also carried out.
Numerical results and Monte-Carlo simulations corroborate the validity of our expressions and show that the SP--plus--CCRT technique helps to greatly reduce the PFA compared to previous studies, thereby improving radar detection.
\end{abstract}


\IEEEpeerreviewmaketitle

\section{Introduction}
\label{sec:Introduction}
One important concern in modern pulsed radars is related to the Doppler frequency estimation of fast-velocity targets.
Due to the high target's radial velocity, ambiguous estimates are more likely to occur. More specifically, ambiguous estimates appear whenever the target's Doppler shift is greater than the pulse repetition frequency (PRF)~\cite{Morris96}.
It seems obvious to think that increasing the PRF will overcome this problem. However, if we are interested in detecting targets located at long distances, then the PRF will be restricted to a maximum value. 
Therefore, the choice of PRF is a trade-off between range and Doppler requirements~\cite{richards10}.
Fortunately, there are some techniques that can resolve ambiguities, although at the cost of extra measurement time and processing load.
These techniques make use of multiples PRFs~\cite{trunk78,hovanession76,Xia07,XLi09,Wang10}.
The most known and used technique is the classic Chinese remainder theorem (CCRT). The CCRT is a fast and accurate method to resolve the unambiguous Doppler frequency. This is accomplished by solving a set of congruences, formed by the estimated measurements of each PRF~\cite{Wang10,Trunk94,Garcia19USA}.
Nevertheless, in this method, the number of PRFs will not be sufficient to resolve a certain quantity of targets. In general, $L$ PRFs are required to successfully disambiguate $L-1$ targets. If the number of targets exceeds $L-1$, then \textit{ghosts} can appear.\footnote{\textit{Ghosts} are false targets resulting from false coincidences of Doppler-ambiguous or range-ambiguous data~\cite{hovanession76}.}
Unless additional data (e.g., tracking information) is available, the radar has no way of recognizing possible false detections~\cite{hovanession76}. 
Care must be taken in the analysis and design since the number of PRFs and the number of targets to be detected have a direct relationship.

Another issue concerning high-velocity targets is related to the signal-to-noise ratio (SNR) loss.
This occurs because the Doppler shift of fast-moving targets will provoke a mismatch between the received signal and its replica~\cite{richards10}.
Consequently, the SNR after \textit{range compression} may be drastically reduced.\footnote{\textit{Range compression} refers to the convolution operation between the received signal and the replica of the transmitted signal~\cite{skolnik01}.}
Some radar systems estimate and remove the Doppler shift prior to applying \textit{range compression}. Nonetheless, some residual or uncompensated Doppler typically remains.
This concern was partially alleviated in~\cite{beltrao17,Barreto19}. Specifically, in~\cite{beltrao17}, the authors proposed a subpulse processing (SP), which proved to have a higher Doppler tolerance,\footnote{Doppler tolerance refers to the degree of degradation in the compressed response due to uncompensated Doppler~\cite{Doviak93}.} increasing the ability to detect fast-moving targets. The shortcomings of SP are computation time (critical for most radars), processing load, and poor velocity resolution. 

As stated before, the CCRT and SP have hardware and physical limitations when it comes to estimating high target velocities. In practice, modern pulsed radars take advantage of these two techniques so as to improve the system's capability to accurately detect the target’s true Doppler frequency.
Since SP the CCRT are affected by the presence of background noise, then a thorough statistical analysis involving these two estimation techniques must be carried out.
Recently in~\cite{silva18}, the authors proposed a novel expression for the probability to correctly estimate the unambiguous Doppler frequency considering the CCRT and the common pulse processing  (PP) technique~\cite{richards10}.
However, to the best of our knowledge, there is no performance analysis considering the SP--plus--CCRT technique.

The main objective of this research is to combine the statistical analysis conducted in~\cite{silva18} along with the newly introduced SP and the CCRT.
To do so, we adopt a stochastic model that suits our Doppler estimation techniques.
Then, we derive novel and closed-form expressions for: \textbf{i)} the probability to correctly estimate the Doppler frequency, also called probability of detection (PD); and \textbf{ii)} the probability to erroneously estimate the Doppler frequency, also called probability of false alarm (PFA). 

The remainder of this paper is organized as follows. Section~\ref{sec: Preliminaries} introduces some key concepts to understand how the velocity estimation is performed. Section~\ref{sec: System Description} describes the system model. Section~\ref{sec: Doppler Detection} analyzes the Doppler estimation using multiple PRFs; Section~\ref{sec: Numerical Results} discusses the representative numerical results. Finally, Section~\ref{sec: Conclusion} concludes this paper.

In what follows, $(a)  \text{mod}  (b)$ denotes the remainder of the euclidean division of $a$ by $b$; $\left| \cdot \right|$, absolute value; $\lfloor\cdot \rfloor$, floor operation; $\text{round} (\cdot)$, rounding operation; $\text{Pr}\left[\cdot \right]$, probability; $\mathbb{E}(\cdot)$, expectation; $\text{Var}(\cdot)$, variance; $(\cdot)^*$, complex conjugate; $\bigcap$, intersection of events; $\bigcup$, union of events; $\mathcal{N}(\mu,\sigma^2)$ denotes a Gaussian distribution with mean $\mu$ and variance $\sigma^2$; $\mathcal{N}_c(\mu,\sigma^2)$ denotes a complex Gaussian distribution with mean $\mu$ and variance $\sigma^2$, and $j= \sqrt{-1}$ is the imaginary unit. 

\section{Preliminaries}
\label{sec: Preliminaries}

In this section, we present a brief introduction about the PP and SP techniques. Latter, we describe the basis to understand the CCRT algorithm. Finally, we show how the combined technique SP--plus--CCRT works in order to improve Doppler estimation.

\subsection{Pulse Processing}
\label{sec: Pulse Processing}

PP is the common technique employed by the radar to estimate the target velocity and improve the SNR. In this processing technique, the radar transmits a sequence of $M$ pulses during a coherent processing interval (CPI)~\cite{richards14}. Then, \textit{range compression} is performed on each pulse to improve the radar's range resolution. Finally, the discrete Fourier transform (DFT) is applied along the slow-time samples to increase the SNR and to estimate the target Doppler frequency~\cite{richards10}. These samples are collected at a rate equal to the $\text{PRF}$.
The maximum Doppler frequency shift that the radar manages to detect using PP is $\Psi_{max}=\pm \text{PRF}/2$. 
If the target Doppler frequency, $\mathit{f}_{d}$, exceeds this value, then the radar will deliver ambiguous Doppler measurements.
The Doppler frequency shift will be positive for closing targets and negative for receding targets. The target velocity, $\mathit{v}_t$, and its corresponding Doppler shift are related by the following equation~\cite{barton13}:
\begin{align}
    \label{eq: target velocity}
    \mathit{f}_{d} = \frac{2 \mathit{v}_t \mathit{f}_R }{\mathit{c}} = \frac{2 \mathit{v}_t }{ \lambda},
\end{align}
where $\mathit{f}_R$ is the radar's operation frequency, $\mathit{c}$ is the speed of light, and $\lambda$ is the radar frequency.

\subsection{Subpulse Processing}
\label{sec: Subpulse Processing}
SP improves Doppler tolerance by mitigating the loss in SNR caused by the uncompensated Doppler shift of fast-moving targets~\cite{richards10,beltrao17}. Moreover, SP is used to overcome the problem of ambiguous Doppler measurements. 
The SP algorithm runs as follows:
\begin{enumerate}
    \item First, the replica of the transmitted signal is divided into $N$ subpulses -- unlike PP that used the entire replica. 
    \item Latter, \textit{range compression} is carried out between each subpulse and the received signal (cf.~\cite{beltrao17,Barreto19} for a detailed discussion on this). The use of shorter replicas will enhance the system's Doppler tolerance~\cite{skolnik01}, increasing the detection capability of fast-moving targets. Of course, this process leads to a reduction in the peak amplitude of the sub-compression response (by a factor of $1/N$).
    Here, the slow-time samples are collected at a rate of $\Phi=N / \tau$, where $\tau$ is the pulse width. 
    It is important to emphasize that PP and SP are performed simultaneously, that is, for each of the $M$ compressions, the radar carried out $N$ sub-compressions~\cite{Barreto19}. 
    \item Finally, the slow-time samples are coherently integrated to estimate the target Doppler frequency and to ``restore'' the peak amplitude of the sub-compression response.
\end{enumerate}
The number of subpulses can be chosen as high as needed, as long as it is taken into consideration that each additional subpulse requires an extra \textit{range compression} operation, increasing the computational load and computation time.
The maximum Doppler frequency shift that the radar can now manage to detect is $\Phi_{max}=\pm  N / 2 \tau$~\cite{beltrao17}.
Since $\Phi_{max}>\Psi_{max}$, SP provides a higher frequency range of detection for fast-moving targets. 
    
Computation time is critical for most radars and depends strongly on the radar's operation mode (e.g. tracking, searching or imaging), thereby limiting the number of subpulses. Commonly, the number of subpulses is set between 5 and 10. However, this small number yields to a poor discretization in the frequency domain and, consequently, producing inaccurate estimates.
%=========================================================
\begin{figure*}
\centering
    \includegraphics[trim={0cm 0cm 0cm 0cm}, clip,scale=0.52]{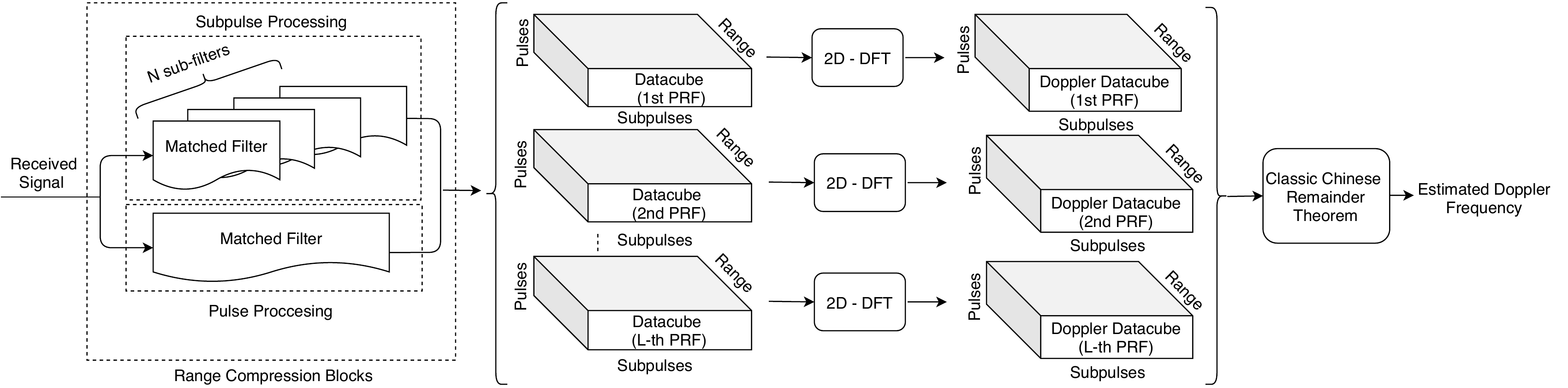} 
  \caption{Block diagram for Doppler estimation.}
  \label{fig: Block Diagram}
  \hrulefill
\end{figure*}
%=========================================================
\subsection{Classic Chinese Remainder Theorem}
\label{sec:Chinese Remainder Theorem}

The use of multiples PRFs is a common approach to resolve range and Doppler ambiguities~\cite{trunk78,Trunk94,trunk93,hovanession76}. In this work, we only focus on solving Doppler ambiguities. Consider for the moment a target with Doppler shift $\mathit{f}_{d}>\Psi_{max}$. In this scenario, the radar will detect the target with an apparent Doppler shift, $\mathit{f}_{d_{ap}}$, that satisfies
\begin{align}
    \label{eq: fd_app}
    \mathit{f}_{d} = \mathit{f}_{d_{ap}} + n \text{PRF},
\end{align}
where $n$ is some integer. It is convenience to express the target's Doppler shift $\mathit{f}_{d}$ in terms of its corresponding Doppler bin, $b_{d}$. Thus,~\eqref{eq: fd_app} becomes
\begin{equation}
    \label{eq: Doppler bin}
    b_{d} = b_{ap}+ n M,
\end{equation}
in which $b_{{ap}} \in \left\{0,1,2,\hdots,M-1\right\}$ is the apparent Doppler bin, defined as
\begin{align}
    b_{ap}=&\left\lfloor   \abs{\frac{\mathit{f}_{d_{ap}}}{\Delta D }} \right\rfloor, \ \ \ \ \ \ \ \ \ \ \ \  \mathit{f}_{d_{ap}} \geq 0 \\
    b_{ap}=& M - \left\lfloor   \abs{\frac{\mathit{f}_{d_{ap}}}{\Delta D }} \right\rfloor, \ \ \ \ \ \   \mathit{f}_{d_{ap}} < 0
\end{align}
with $\Delta D=\text{PRF}/M$ being the Doppler bin spacing.
Under this scenario, the radar is incapable to detect the target's true Doppler frequency.

Now, suppose that we have $L$ PRFs. 
Then, the unambiguous target's Doppler bin must satisfies the following congruences:
\begin{align}
    \label{eq: congruences}
    b_{d} \equiv b_{{ap}_i} + n_i M_i, \ \ \ \ 1\leq i \leq L
\end{align}
The CCRT states that if all PRFs are pairwisely coprimes, then the set of congruences in~\eqref{eq: congruences} will have a unique solution given by~\cite{trunk93,hovanession76,ferrari97}
\begin{align}
    \label{}
    b_d = \left( \sum_{i=1}^{L} b_{{ap}_i} \beta_i \right)  \text{mod}  \left( \Theta \right),
\end{align}
where $\Theta = \prod_{i=1}^{L} M_i$, $\beta_i = b_i \Theta/ \text{PRF}_i$, and $ b_i$ is the smaller integer which can be computed by solving the following expression:
\begin{align}
    \label{}
    \left(  \frac{b_i \Theta}{M_i} \right)\text{mod} \left( M_i\right) =1.
\end{align} 

\subsection{Doppler Estimation}
\label{sec: Doppler Detection}

Fig.~\ref{fig: Block Diagram} depicts the entire block diagram for Doppler estimation.
First, the received signal passes through two types of independent \textit{range compression} blocks, one for PP and one for SP. This process is performed in sequence for each pulse repetition interval (PRI).
The outputs of both blocks are combined and stored in memory to form a datacube~\cite{richards10}. (The datacube's data is organized by range, number of pulses, and number of subpulses.) More datacubes are needed when using more than one PRF, as shown in Fig.~\ref{fig: Block Diagram}.
Next, a 2D-DFT block is applied to each datacube to perform coherent integration. (The 2D-DFT block is referred to as a two-dimensional DFT applied along with pulses and subpulses.) 
Latter, the output of the 2D-DFT block is a matrix with the same size containing the estimated Doppler shifts. This new matrix is referred to as Doppler datacube.
Finally, the CCRT is applied over the Doppler datacubes. 
This process will be clarified in Section~\ref{sec: Numerical Results} by means of simulation. 

Noise, jammer, and clutter are major concerns in all radar systems.
In this work, we consider the presence complex white Gaussian noise (CWGN). Thus, the Doppler spectrum of fast moving targets will be compromised due to the intrinsic characteristics of noise. 
For example, a high noise power could mask small target returns, degrading radar performance.
Even if the target return is entirely deterministic, the combined signal (target--plus--noise) is a random process and must be treated as such.
Therefore, we need to assess the statistics underlying Doppler analysis, but first, we need to come up with a specific stochastic model that suits the requirements and design of our radar's estimation scheme. This is discussed in the next section.

\section{System Model}
\label{sec: System Description}
In this section, we propose a stochastic model that fits our signal processing schemes. In addition, we describe the premises (hypotheses) used for Doppler estimation.

According to Sections~\ref{sec: Pulse Processing} and \ref{sec: Subpulse Processing}, the collected signals in the slow-time domain corresponding to PP and SP can be expressed, respectively, as
\begin{align}
    \label{eq:g1}
    \nonumber g_1 \left[m\right] = & s_1 \left[m\right] + w_1 \left[m\right]\\
    = & a_1\exp \left(j 2 \pi \mathit{f}_{d}  m / \text{PRF} \right) + w_1 \left[m\right], \ \  0\leq m \leq M-1 \\ \label{eq:g2}
    \nonumber g_2 \left[n\right] = & s_2 \left[n\right] + w_2 \left[n\right]\\
    = & a_2 \exp \left(j 2 \pi \mathit{f}_{d}  n / \Phi\right) + w_2 \left[n\right], \ \ \ \ \ \ \ 0\leq n \leq N-1
\end{align}
where $s_1 \left[m\right]$ and $s_2 \left[n\right]$ are discrete complex sine signals\footnote{In most systems, the radio frequency (RF) signal is mixed to baseband prior to compression, and a coherent detector is used in the downconversion process to form in-phase (I) and quadrature (Q) receive channels, thereby creating a complex baseband signal.} originated by changes in the target position; $w_1 \left[m\right]$ and $w_2 \left[n\right]$ are discrete additive complex Gaussian noises; and finally, $a_1$ and $a_2$ are the amplitudes at the output of the matched filters.
Depending on the target velocity, the output amplitudes $a_1$ and $a_2$ maybe be greatly attenuated. However, the attenuation in $a_2$ is partially mitigated by the use of SP. In particular, it follows that $a_2>a_1$ for high-velocity targets~\cite{beltrao17}.
Additionally, we define $2\sigma_{t_1}^2$ and $2\sigma_{t_2}^2$ as the total  mean powers -- in the time domain -- for $w_1 \left[m\right]$  and $w_2 \left[n\right]$, respectively. 
As seen in practice, and due to the stationary characteristic of noise, we have that $\sigma_{t_1}^2 =\sigma_{t_2}^2$~\cite{papoulis02}. However, we will remain using separate notations for $\sigma_{t_1}^2$ and $\sigma_{t_2}^2$ so as to distinguish the noise power from PP and SP.
Of course, these separate notations will not alter, in any form, our performance analysis.

The SNR measured in the time domain considering PP and SP, can be expressed, respectively,~as
\begin{align}
    \label{eq: SNR Time 1}
    \text{SNR}_{t_1}= &\frac{\left| a_{1} \right|^{2}}{ 2\sigma_{t_1}^{2}} \\ \label{eq: SNR Time 2}
    \text{SNR}_{t_2}= &\frac{\left| a_{2} / N \right|^{2}}{2 \sigma_{t_2}^{2}}.
\end{align}
Observe in \eqref{eq: SNR Time 2} that the fact of dividing the replica into $N$ subpulses causes a reduction in the SNR by a factor of $1/N^2$, as mentioned in Section \ref{sec: Subpulse Processing}.

The DFT is the primary operation to implement coherent integration. More precisely, the DFT provides a mechanism to test multiple candidate frequencies to maximize the \textit{integration gain}~\cite{richards10}.
The corresponding DFTs for \eqref{eq:g1} and \eqref{eq:g2} are given, respectively, by
\begin{align}
    \label{eq:DFT G1}
    \nonumber G_1\left[ k' \right] \triangleq & \  \mathscr{F} \left\{ g_1 \left[m\right] \right\} \\
    \nonumber =& \sum_{m=0}^{M-1} g_1 \left[m\right] \exp \left(-j 2 \pi k' m / M \right)\\
    = & S_1\left[ k' \right] + W_1 \left[ k' \right], \ \ \ 0\leq k' \leq M-1 \\ \label{eq:DFT G2}
    \nonumber G_{2}\left[ l' \right] \triangleq & \ \mathscr{F} \left\{ g_2 \left[n\right] \right\} \\
    \nonumber =& \sum_{n=0}^{N-1} g_2 \left[n\right] \exp \left(-j 2 \pi l' n / N \right)\\
    = & S_2\left[ l' \right] + W_2 \left[ l' \right], \ \ \ \ \  \ 0\leq l' \leq N-1
\end{align}
The SNR measured in the frequency domain considering PP and SP, are given, respectively, by~\cite[Eq. (17.37)]{richards10}
\begin{align}
    \label{eq: SNR Freq 1}
    \text{SNR}_{1} =& \frac{| M a_1 |^2}{2 \sigma_{1}^{2} } \\ \label{eq: SNR Freq 2}
    \text{SNR}_{2} =& \frac{|a_2 |^2}{2 \sigma_{2}^{2}},
\end{align}
in which $\sigma_{1}^{2}= M \sigma_{t_1}^{2}$ and $\sigma_{2}^{2}=N \sigma_{t_1}^{2}$ are half of the noise powers -- in the frequency domain -- for $W_1 \left[ k' \right]$ and $W_2 \left[ l' \right]$, respectively.

The Doppler estimates are based on the absolute values of $G_1\left[ k' \right]$ and $G_2\left[ l' \right]$. That is,~\eqref{eq:DFT G1} and \eqref{eq:DFT G2} will provide estimates for $\mathit{f}_{d}$, say $\hat{\mathit{f}}_{1}$ and $\hat{\mathit{f}}_{2}$, by searching $k'$ and $l'$, in which the absolute values of $G_1\left[ k' \right]$ and $G_2\left[ l' \right] $ are maximum.
It is worth mentioning that if $\Psi_{max}<\mathit{f}_{d}$ and $\Phi_{max}<\mathit{f}_{d}$, then
$\hat{\mathit{f}}_{1}$  and $\hat{\mathit{f}}_{2}$ will display ambiguous Doppler estimates.

Now, considering $L$ PRFs (say, $\text{PRF}_1,\ldots,\text{PRF}_L$), we can define the absolute values for~$G_1\left[ k' \right]$ and $G_2\left[ l' \right]$ at the $i$-th PRF, respectively,~as
\begin{align}
    \label{eq:X1i}
    H_{1,i} \left[k'\right] \triangleq &| G_{1,i}\left[ k' \right] | \ \ \ \ \  0\leq k' \leq M_i-1  \\ \label{eq:X2i}
    H_{2,i} \left[l'\right] \triangleq & | G_{2,i}\left[ l' \right] | \ \ \ \ \ \  0\leq l' \leq N_i-1
\end{align}
where the subscript $i \in \{1,\ldots,L\}$ denotes the association to the $i$-th PRF.

Herein, we assume that $G_{1,i} \left[ k' \right]$ is composed of $M_i-1$ independent and identically distributed noise samples and one target--plus--noise sample, denoted as $\mathcal{G}_{1,i}$. On the other hand,  $G_{2,i} \left[ l' \right]$ is composed of $N_i-1$ independent and identically distributed noise samples and one combined sample, denoted as $\mathcal{G}_{2,i}$.
The target--plus--noise samples $\mathcal{G}_{1,i}$ and $\mathcal{G}_{2,i}$ can be modeled, respectively, by~\cite[Eq. (1)]{beaulieu11}
\begin{align}
    \label{eq:G1}
    \nonumber \mathcal{G}_{1,i} = &\sigma_{1,i} \left(\sqrt{1- \lambda_{1,i}^2} A_{1,i} +\lambda_{1,i} A_{0,i} \right) \\
    & + j \sigma_{1,i} \left(\sqrt{1-\lambda_{1,i}^2} B_{1,i} + \lambda_{1,i} B_{0,i} \right) \\ \label{eq:G2}
    \nonumber \mathcal{G}_{2,i} = &\sigma_{2,i} \left(\sqrt{1- \lambda_{2,i}^2} A_{2,i} +\lambda_{2,i} A_{0,i} \right) \\
    & + j \sigma_{2,i} \left(\sqrt{1-\lambda_{2,i}^2} B_{2,i} + \lambda_{2,i} B_{0,i} \right),
\end{align}
where $A_{p,i}$ and $B_{p,i}$ ($p=1,2$) are mutually independent random variables (RVs) distributed as~$\mathcal{N}(0,\frac{1}{2})$, and $ \lambda_{p,i}, \in (0,1]$.
Then, for any $p$ and $q$ ($q=1,2$), it follows that $\mathbb{E} (A_{p,i} B_{q,i})=0$ and $\mathbb{E} (A_{p,i}A_{q,i})=\mathbb{E} (B_{p,i} B_{q,i})= \frac{1}{2} \delta_{pq}$. ($\delta_{pq}=1$ if $p=q$, and $\delta_{pq}=0$ otherwise.)
In addition, $A_{0,i}$ and $B_{0,i}$ are mutually independent RVs distributed as $\mathcal{N}(m_{\textbf{Re},i},\frac{1}{2})$ and $\mathcal{N}(m_{\textbf{Im},i},\frac{1}{2})$, respectively.
Thus, $\mathcal{G}_{1,i}$ and $\mathcal{G}_{2,i}$ are non-zero mean complex Gaussian RVs with probability density functions (PDFs) given, respectively, by $\mathcal{N}_c( \lambda_{1,i} (m_{\textbf{Re},i}+j m_{\textbf{Im},i}) ,\sigma_{1,i}^2)$ and $\mathcal{N}_c( \lambda_{2,i} (m_{\textbf{Re},i}+j m_{\textbf{Im},i}) ,\sigma_{2,i}^2)$. 
The correlation coefficient between any pair of ($\mathcal{G}_{1,i}$, $\mathcal{G}_{2,i}$), can be calculated as~\cite[Eq. (2)]{beaulieu11}
\begin{align}
    \label{eq:}
    \nonumber \rho_{kl,i} \triangleq & \frac{\mathbb{E} ( \mathcal{G}_{1,i} \mathcal{G}_{2,i}^*) - \mathbb{E}(\mathcal{G}_{1,i}) \mathbb{E}(\mathcal{G}_{2,i}^*)}{\sqrt{\text{Var}(\mathcal{G}_{1,i}) \text{Var}(\mathcal{G}_{2,i})}} \\
    = & \lambda_{1,i} \lambda_{2,i}.
\end{align}
This correlation exists because both PP and SP use the same received signal when performing \textit{range compression}~\cite{richards10}. 
Observe that the parameters $\lambda_{1,i}^2$, $\lambda_{2,i}^2$, $m_{\textbf{Re},i}$ and $m_{\textbf{Im},i}$ can be used to model the compressed responses $\left| M_i a_{1,i} \right|^{2}$ and $\left|a_{2,i} \right|^{2}$.
This can be done by making the following substitutions: $|M_i a_{1,i}|^2=\lambda_{1,i}^2 (m_{\textbf{Re},i}^2+m_{\textbf{Im}}^2)$ and $|a_{2,i}|^2=\lambda_{2,i}^2 (m_{\textbf{Re},i}^2+m_{\textbf{Im}}^2)$.
On the other hand, $\lambda_{1,i}$ and $\lambda_{2,i}$ can be chosen to meet a desire correlation coefficient.

By the above, it follows that $H_{1,i} \left[ k' \right]$ is composed of $M_i-1$ Rayleigh distributed samples, denoted as $X_{k,i}$ $\left(k\in \left\{ 1,2,\hdots,M_i-1 \right\} \right)$, and one Rice distributed sample, denoted as $R_{1,i}$. Similarly, $H_{2,i} \left[ l' \right]$ is composed of $N_i-1$ Rayleigh distributed samples, denoted as $Y_{l,i}$ $\left(l \in \left\{ 1,2,\hdots,N_i-1 \right\} \right)$, and one Rice distributed sample, denoted as $R_{2,i}$. 
The PDFs of $X_{k,i}$ and $Y_{l,i}$ are given, respectively, by 
\begin{align}
    \label{eq:PDF x1}
    f_{X_{k,i}}(x_{k,i}) = & \frac{ x_{k,i} \exp\left(- \frac{ x_{k,i}^2}{2\sigma_{k,i}^2} \right)}{\sigma_{k,i}} \\ \label{eq:PDF x2} 
    f_{Y_{l,i}}(y_{l,i}) = & \frac{y_{l,i} \exp\left(- \frac{y_{l,i}^2}{2 \sigma_{l,i}^2} \right)}{\sigma_{l,i}}.
\end{align}
Moreover, since $R_{2,i}$ and $R_{2,i}$ bear a certain degree of correlation, they are governed by a bivariate Rician distribution, given by~\cite{beaulieu11,Behnad12}
\begin{align}
    \label{eq:PDF x1 x2}
    \nonumber \mathit{f}_{R_{1,i},R_{2,i}}& \left( r_{1,i},r_{2,i}  | \mathcal{H}_1\right) = \int_{0}^{\infty}   \exp\left( -t \xi_i\right) \\
    \nonumber &  \times \exp\left(-\textbf{m}_i\right) I_0\left( 2 \sqrt{\textbf{m}_it} \right) \prod_{p=1}^{2} \frac{r_{p,i}}{\Omega_{p,i}^{2}}  \\
    & \times \exp\left(- \frac{r_{p,i}^2}{2 \Omega_{p,i}^2} \right)  I_0 \left( \frac{r_{p,i} \sqrt{t \sigma_{p,i}^2 \lambda_{p,i}^2}}{\Omega_{p,i}^2}\right) \text{d}t,
\end{align}
where $ I_0(\cdot )$ is the modified Bessel function of the first kind and order zero~\cite[Eq. (9.6.16)]{abramowitz72}, $\textbf{m}_i =m_{\textbf{Re},i}^2 +m_{\textbf{Im},i}^2$, and 
\begin{subequations}
\begin{align}
    \label{eq:Omega}
    \Omega_{p,i}^2 &= \sigma_{p,i}^2 \left( \frac{1 - \lambda_{p,i}^2}{2}\right) \\
    \xi_{i} & = 1+ \sum_{p=1}^{2} \frac{\sigma_{p,i}^2 \lambda_{p,i}^2}{2 \Omega_{p,i}^2}.
\end{align}
\end{subequations}

\section{Doppler Analysis}
\label{sec: Doppler Detection}
In this section, we provide a comprehensive statistical analysis on Doppler estimation.
To do so, we derive the performance metrics for both SP and SP--plus--CCRT.

\subsection{SP Analysis}
\label{sec:Probability of Detection by PRF}
First, let us define the following events:
\begin{align}
    \label{}
    \mathcal{A}_{k,i}=& \left\{R_{1,i}>X_{k,i} \right\} \\
    \mathcal{B}_{l,i}=&\left\{R_{2,i}>Y_{l,i} \right\} \\
    \mathcal{C}_{k,i}=& \left\{X_{k,i}>R_{1,i} \right\} \\
    \mathcal{D}_{l,i}=&\left\{ Y_{l,i}>R_{2,i} \right\}.
\end{align}

\begin{proposition1}
\label{proposition 1}
Let $\text{PD}_i$ be the probability of detection at the $i$-th PRF.
Specifically, $\text{PD}_i$ is defined as the probability that $R_{1,i}$ is greater than $X_{k,i}$ and, simultaneously, that $R_{2,i}$ is greater than $Y_{l,i}$, i.e.,
\begin{align}
    \label{th: proposition 1}
    \text{PD}_i \triangleq \text{Pr} &\left[ \left( \bigcap_{k=1}^{M_i-1} \mathcal{A}_{k,i} \right) \bigcap  \left( \bigcap_{l=1}^{N_i-1} \mathcal{B}_{l,i} \right) \right].
\end{align}
Then, from \eqref{eq:PDF x1}--\eqref{eq:PDF x1 x2}, \eqref{th: proposition 1} can be expressed in closed-form as
\begin{align}
    \label{eq:PD_final}
    \nonumber \text{PD}_i = & \sum _{k=0}^{M_i-1} \sum _{l=0}^{N_i-1} \left(
\begin{array}{c}
 M_i-1 \\ k \\ \end{array}
\right) \left( \begin{array}{c}
 N_i-1 \\ l \\ \end{array} \right) \\
 & \times\frac{ (-1)^{-k-l+M_i+N_i} \mathcal{V}_i(k,l)}{\mathcal{U}_i(k,l) } \exp \left(- \textbf{m}_i + \frac{\textbf{m}_i}{\mathcal{U}_i(k,l)} \right),
\end{align}
wherein $\mathcal{U}_i(k,l)$ and $\mathcal{V}_i(k,l)$ are auxiliary functions defined, respectively, as
\begin{subequations}
\begin{align}
    \label{eq:}
    \nonumber \mathcal{U}_i(k,l)  = & \ \xi_{i}-\frac{\xi_i  \lambda_{1,i}^2 \sigma_{1,i}^4}{2\Omega_{1,i}^2 \left( \Omega_{1,i}^2 (k-M_i+1)-\sigma_{1,i}^2\right)} \\
    & -\frac{\xi_i \lambda_{2,i}^2 \sigma_{2,i}^4}{2\Omega_{2,i}^2 \left( \Omega_{2,i}^2 (l-N_i+1)-\sigma_{2,i}^2\right)}\\
    \nonumber \mathcal{V}_i(k,l) =& \  \frac{\sigma_{1,i}^2}{\left(\Omega_{1,i}^2 (-k+M_i-1)+\sigma_{1,i}^2\right) } \\
    & \times \frac{\sigma_{2,i}^2}{\left( \Omega_{2,i}^2 (-l+N_i-1)+\sigma_{2,i}^2\right)}.
\end{align}
\end{subequations}
\end{proposition1}

\begin{proof}
See Appendix~\ref{app: Derivation 1}.
\end{proof}

\begin{corollary}
\label{corollary 2}
Let $\text{PFA}_i$ be the probability of false alarm at the $i$-th PRF.
More precisely, $\text{PFA}_i$ is defined as the probability that at least one of $X_{k,i}$ is greater than $R_{1,i}$ and, simultaneously, that at least one of $Y_{l,i}$ is greater than $R_{2,i}$, i.e.,
\begin{align}
    \label{th: corollary 1}
     \text{PFA}_i \triangleq                 \text{Pr} &\left[ \underset{k=1}{\overset{M_i-1}{\bigcup }} \underset{l=1}{\overset{N_i-1}{\bigcup }} \left( \mathcal{C}_{k,i} \bigcap \mathcal{D}_{l,i} \right) \right].
\end{align}
Then, from \eqref{eq:PDF x1}--\eqref{eq:PDF x1 x2}, \eqref{th: corollary 1} can be written in closed-form as in \eqref{eq: PFA final}, shown at the top of the next page, where $\mathcal{P}_i\left(k,l \right)$ and $\mathcal{Q}_i\left(k,l \right)$ are auxiliary functions defined, respectively, by
\begin{subequations}
\label{eq: PQ functions}
\begin{align}
    \label{eq: P function}
    \nonumber \mathcal{P}_i\left(k,l \right) = & \  \xi _i-\frac{\lambda _{1,i}^2 \sigma ^4_{1,i}}{2 \Omega _{1,i}^2 \left(k \ \Omega _{1,i}^2+\sigma ^2_{1,i}\right)} \\
    & -\frac{\lambda _{2,i}^2 \sigma ^4_{2,i}}{2 \Omega _{2,i}^2 \left(l \  \Omega _{2,i}^2+\sigma ^2_{2,i}\right)} \\ \label{eq: Q function}
    \mathcal{Q}_i \left(k,l \right) =& \ \frac{\sigma ^2_{1,i} \sigma ^2_{2,i}}{\left(k \ \Omega _{1,i}^2+\sigma ^2_{1,i}\right) \left(l\  \Omega _{2,i}^2+\sigma ^2_{2,i}\right)}.
\end{align}
\end{subequations}
%=========================================================
\begin{figure*}[!t]
%\hrulefill
\begin{flushleft}
\small
\begin{align}
    \label{eq: PFA final}
    \nonumber \textit{PFA}_i = & \frac{\left(M_i-1 \right)  \left(N_i-1 \right)  \mathcal{Q}_i \left(1,1 \right)}{ \mathcal{P}_i \left(1,1 \right) } \exp \left( -\textbf{m}_i+ \frac{\textbf{m}_i}{\mathcal{P}_i \left(1,1 \right)} \right) - \binom{M_i-1}{2} \binom{N_i-1}{2}  \frac{\mathcal{Q}_i \left(2,2 \right)}{ \mathcal{P}_i \left(2,2 \right) } \exp \left( -\textbf{m}_i+ \frac{\textbf{m}_i}{\mathcal{P}_i \left(2,2 \right)} \right)  + \hdots \\
    & + (-1)^{M_i-N_i-1} \frac{\mathcal{Q}_i \left(M_i-1,N_i-1 \right)}{ \mathcal{P}_i \left(M_i-1,N_i-1  \right) } \exp \left( -\textbf{m}_i+ \frac{\textbf{m}_i}{\mathcal{P}_i \left(M_i-1,N_i-1  \right)} \right) 
\end{align}
\normalsize
\end{flushleft}
\hrulefill
\end{figure*}
%=========================================================
\end{corollary}

\begin{proof}
See Appendix~\ref{app: Derivation 2}.
\end{proof}
It is worth mentioning that~\eqref{eq:PD_final} and \eqref{eq: PFA final} are novel and original contributions of this work, derived in closed-form even though~\eqref{eq:PDF x1 x2} is given in integral form.

\subsection{SP--Plus--CCRT Analysis}

Similar to~\cite{silva18}, we assume that each individual pulse on each sweep results in an independent random value for the target returns.

Now, using~\eqref{eq:PD_final} and taking into account the $\mathcal{M}$--of--$L$ detection criterion,\footnote{Instead of detecting a target on the basis of at least one detection in $L$ tries, system designers often require that some number $\mathcal{M}$ or more detections be required in $L$ tries before a target detection is accepted~\cite{richards10}.} the probability of detection for the combined technique SP--plus--CCR can be calculated as follows~\cite{Wang93}
%=========================================================
%=========================================================
\begin{align}
    \label{eq: PD Poisson binomial distribution CCRT}
    \text{PD}_{\text{CCRT}} \triangleq &\sum _{l=\mathcal{M}}^{L} \sum _{\mathcal{E} \in \mathcal{F}_l }  \left\{ \left( \prod_{i \in \mathcal{E}} \text{PD}_i  \right) \left( \prod_{j \in \mathcal{E}^c} \left(1- \text{PD}_j \right) \right) \right\},
\end{align}
\normalsize
where $\mathcal{F}_l$ is the set of all subsets of $l$ integers that can be selected from $\left\{1,2,\hdots,L \right\}$, and $\mathcal{E}^c$ is the complement of $\mathcal{E}$. For example, if $l=2$ and $L=3$, then $\mathcal{F}_2=\left\{\left\{ 1,2 \right\},\left\{ 1,3 \right\},\left\{ 2,3 \right\}  \right\}$, and $\mathcal{E}^c= \left\{1,2,\hdots,L  \right\} \backslash \mathcal{E}$.

On the other hand, the probability of false alarm for the combined technique SP--plus-CCRT can be calculated as~\cite{Wang93}
\begin{align}
    \label{eq: PFA Poisson binomial distribution CCRT}
    \text{PFA}_{\text{CCRT}} \triangleq &\sum _{l=\mathcal{M}}^{L} \sum _{\mathcal{E} \in \mathcal{F}_l }  \left\{ \left( \prod_{i \in \mathcal{E}} \text{PFA}_i  \right) \left( \prod_{j \in \mathcal{E}^c} \left(1- \text{PFA}_j \right) \right) \right\}.
\end{align}
For the case where $\mathcal{M} = L$, \eqref{eq: PD Poisson binomial distribution CCRT} and \eqref{eq: PFA Poisson binomial distribution CCRT} reduce, respectively, to
\begin{align}
    \label{eq:PD_product}
    \text{PD}_{\text{CCRT}} = &\prod_{i=1}^{L} \text{PD}_i  \\ \label{eq:PFA_product}
    \text{PFA}_{\text{CCRT}} = &\prod_{i=1}^{L} \text{PFA}_i.
\end{align} 

\section{Numerical Results}
\label{sec: Numerical Results}
%=========================================================
\begin{figure}[t!]
\begin{center}
\includegraphics[trim={0cm 0cm 0cm 0cm}, clip, scale=0.4]{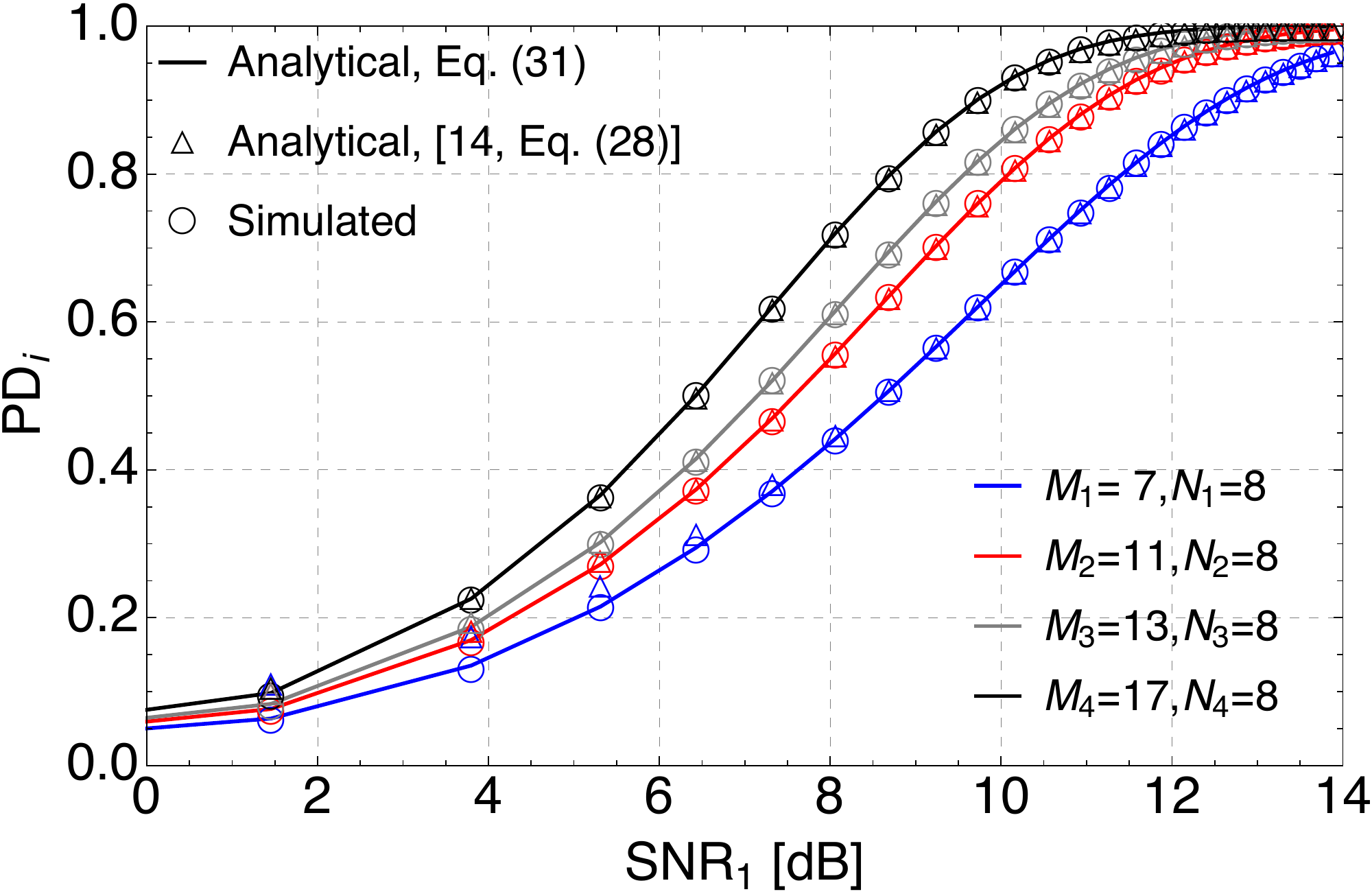}
\caption{$\text{PD}_i$ vs $\text{SNR}_1$ using $N_i=8$, $\lambda_{1,i}=0.5$, $\lambda_{2,i}=0.99$, and different values of $M_i$ ($i \in \left\{1,2,3,4 \right\}$).}
\label{fig:PD1vsrho} 
\end{center}
\end{figure}
%=========================================================
% %=========================================================
\begin{figure}[t!]
\begin{center}
\includegraphics[trim={0cm 0cm 0cm 0cm}, clip, scale=0.4]{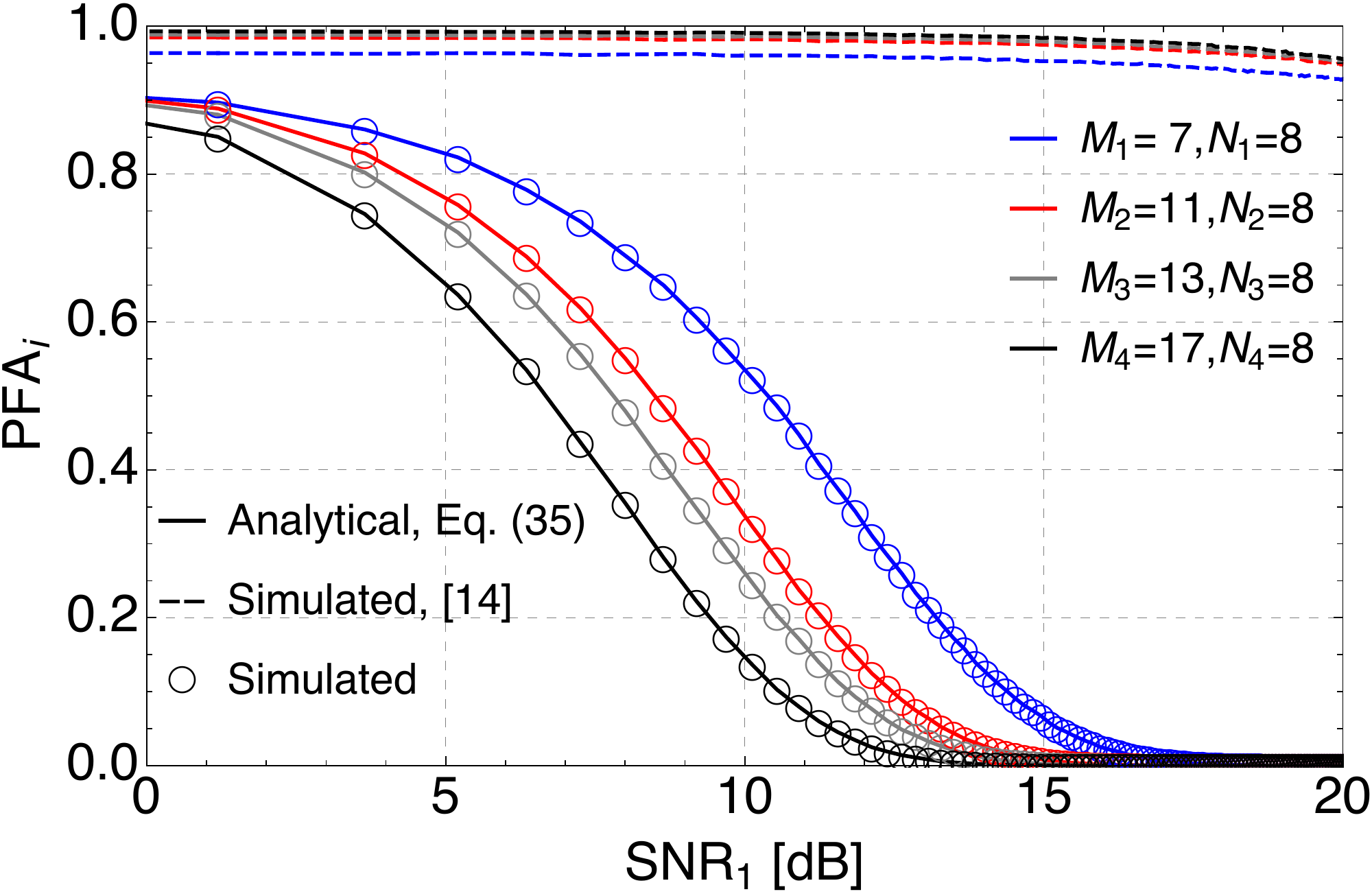}
\caption{$\text{PMD}_i$ vs $\text{SNR}_1$ using $N_i=8$, $\lambda_{1,i}=0.5$, $\lambda_{2,i}=0.99$, and different values of $M_i$ ($i \in \left\{1,2,3,4 \right\}$).}
\label{fig: PD1vsm1} 
\end{center}
\end{figure}
%=========================================================
%=================VALE========================================
\begin{figure*}
  \centering
  \begin{tabular}[b]{c}
  \includegraphics[trim={1cm 6.5cm 4cm 6.5cm}, clip,scale=0.25]{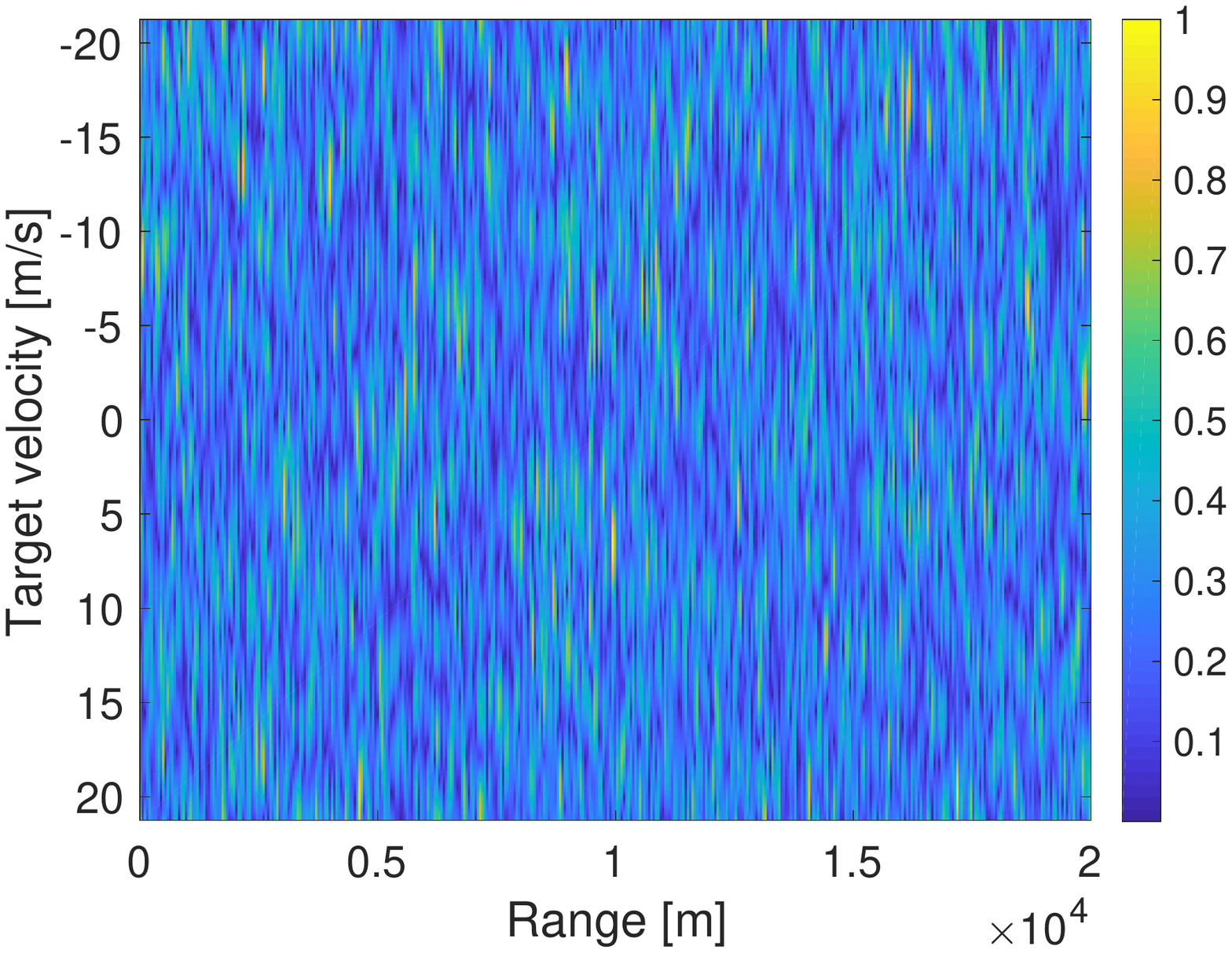} \\
    \scriptsize (a--1) $\text{PRF}_1=1700$ [Hz], $M_1=11$ 
  \end{tabular} \hspace{-1.2cm} \qquad 
  \begin{tabular}[b]{c}
    \includegraphics[trim={1cm 6.5cm 4cm 6.5cm}, clip,scale=0.25]{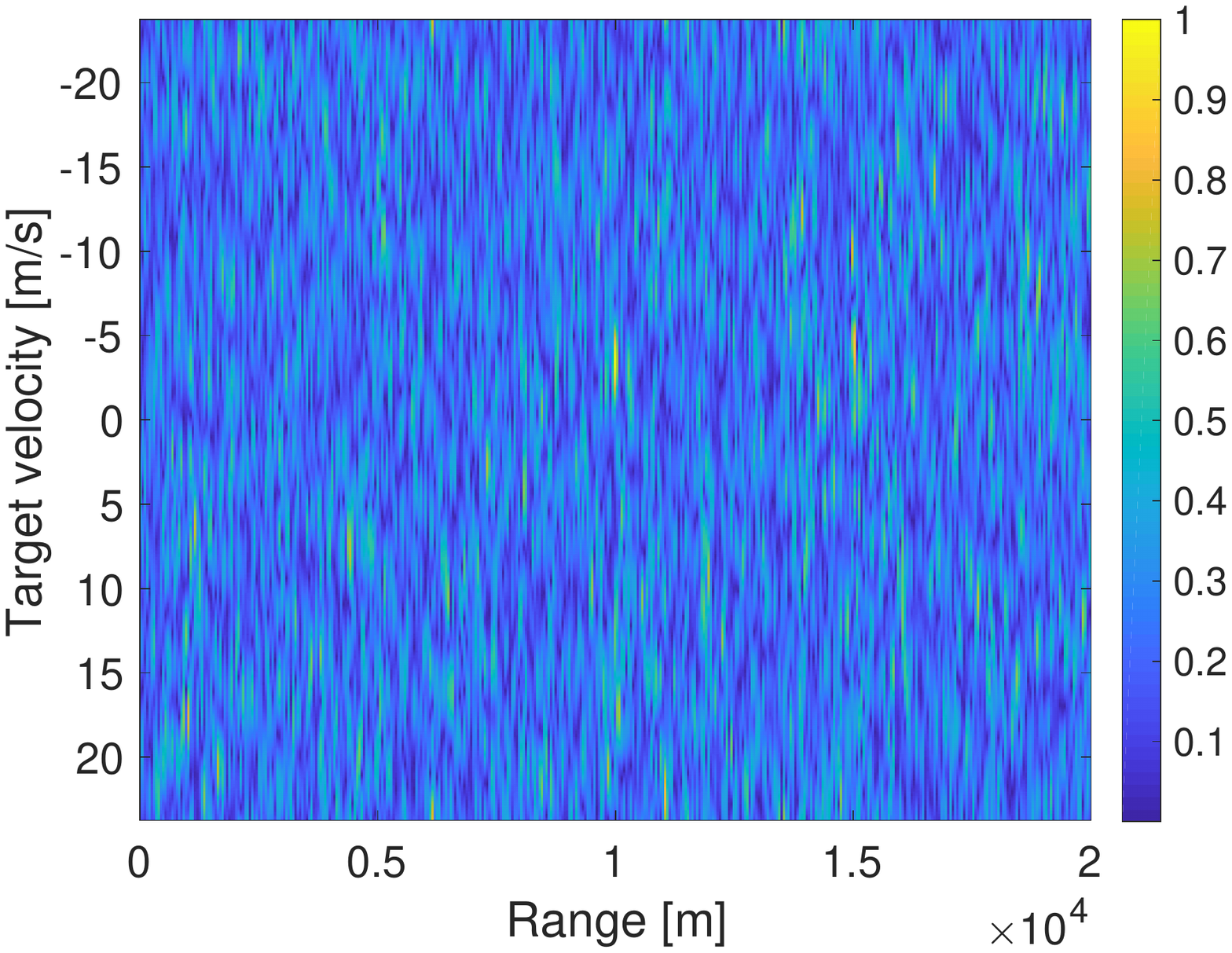} \\
    \scriptsize (a--2) $\text{PRF}_2=1900$ [Hz], $M_2=13$ 
  \end{tabular}  \hspace{-1.2cm}  \qquad
  \begin{tabular}[b]{c}
    \includegraphics[trim={1cm 6.5cm 4cm 6.5cm}, clip,scale=0.25]{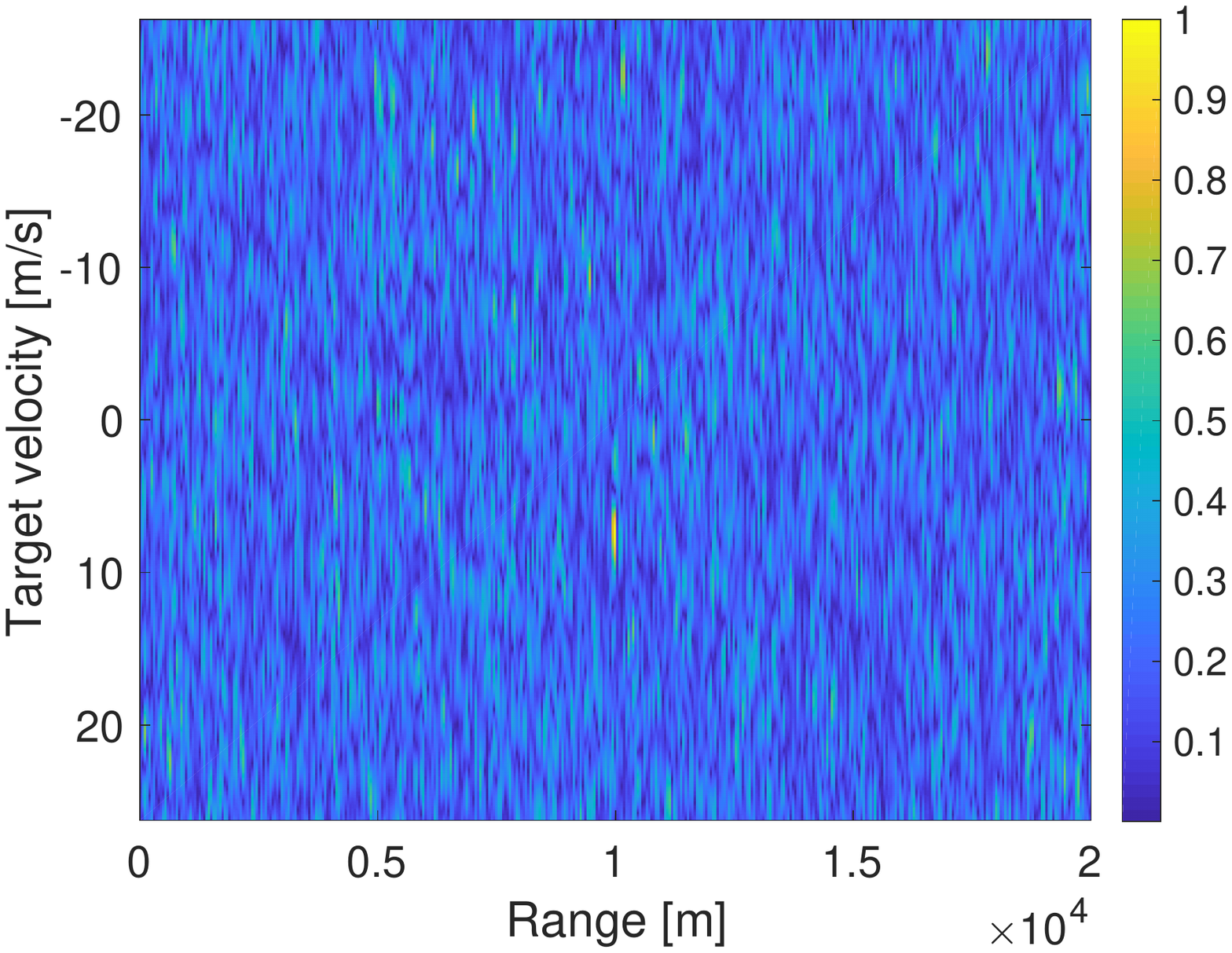} \\
    \scriptsize (a--3) $\text{PRF}_3=2100$ [Hz], $M_3=17$ 
  \end{tabular} \hspace{-1.2cm} \qquad 
  \begin{tabular}[b]{c}
    \includegraphics[trim={1cm 6.5cm 1.5cm 6.5cm}, clip,scale=0.25]{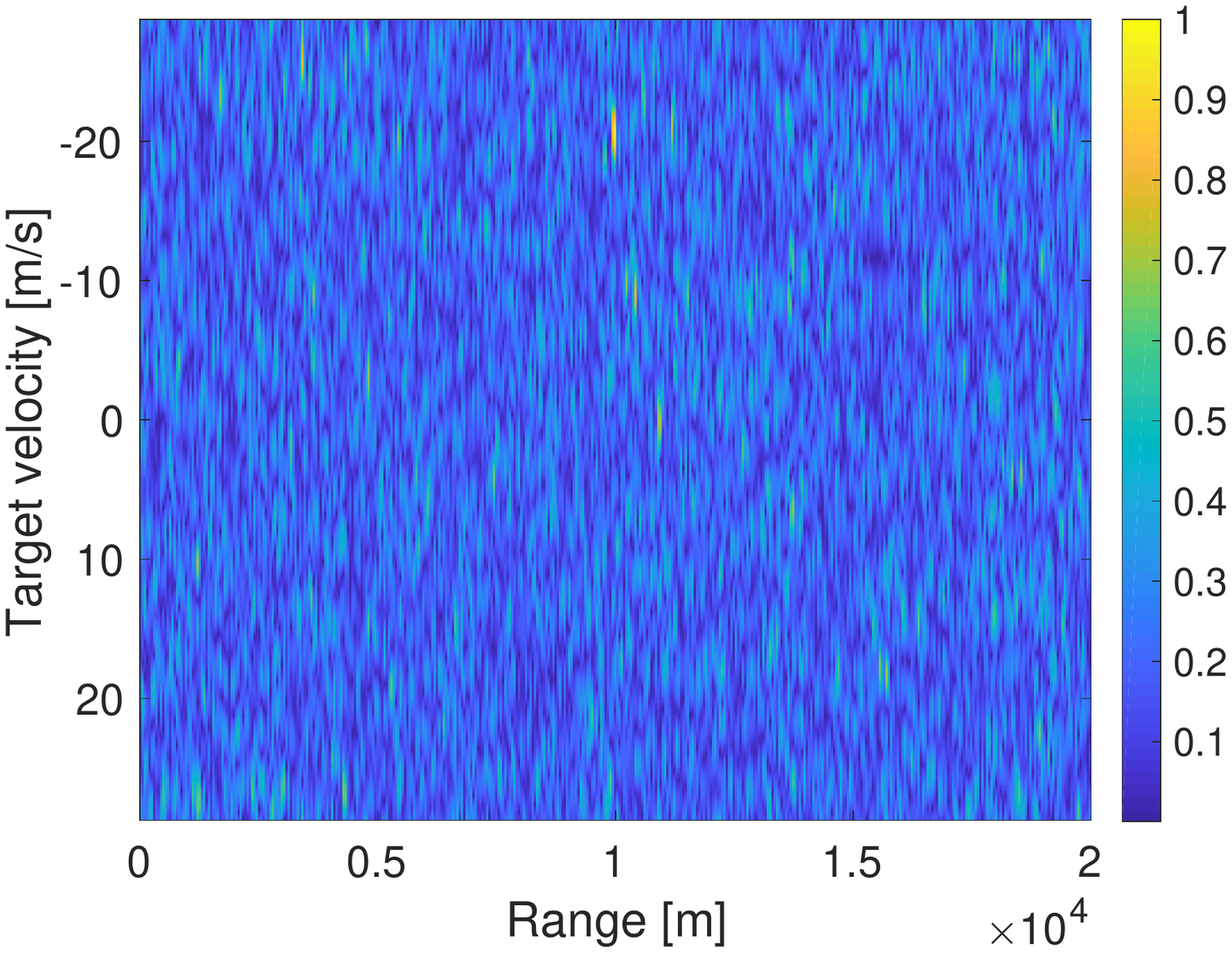} \\
    \scriptsize (a--4) $\text{PRF}_4=2300$ [Hz], $M_4=19$ 
  \end{tabular}
  %%%%%%%%%%%%%%%%%%%%%%%%%%%%%%%%%%%%%%%%%%%%%%%%
    \begin{tabular}[b]{c}
  \includegraphics[trim={1cm 6.5cm 4cm 6.5cm}, clip,scale=0.25]{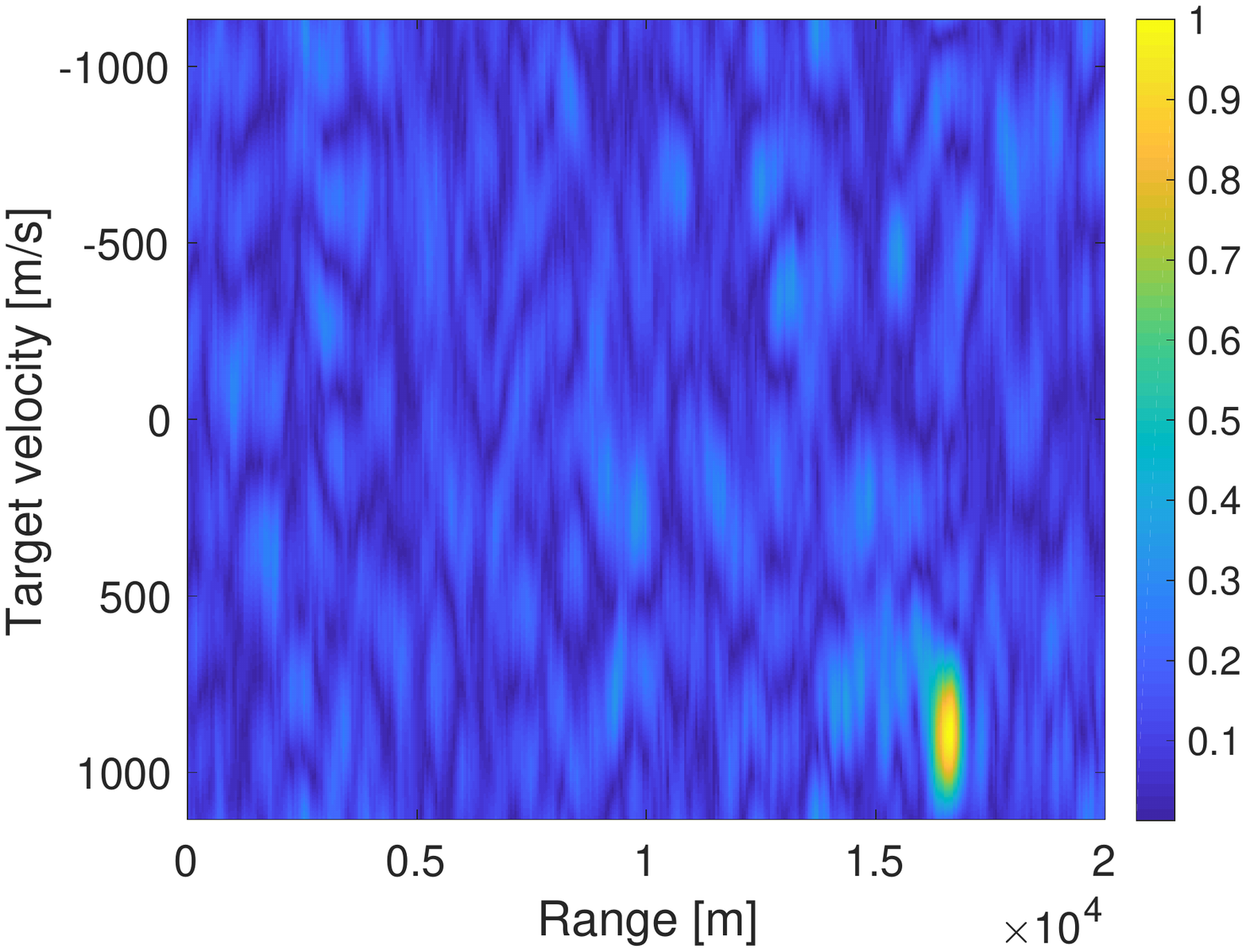} \\
    \scriptsize (b--1) $\text{PRF}_1=1700$ [Hz], $N_1=8$
  \end{tabular} \hspace{-1.2cm} \qquad 
  \begin{tabular}[b]{c}
    \includegraphics[trim={1cm 6.5cm 4cm 6.5cm}, clip,scale=0.25]{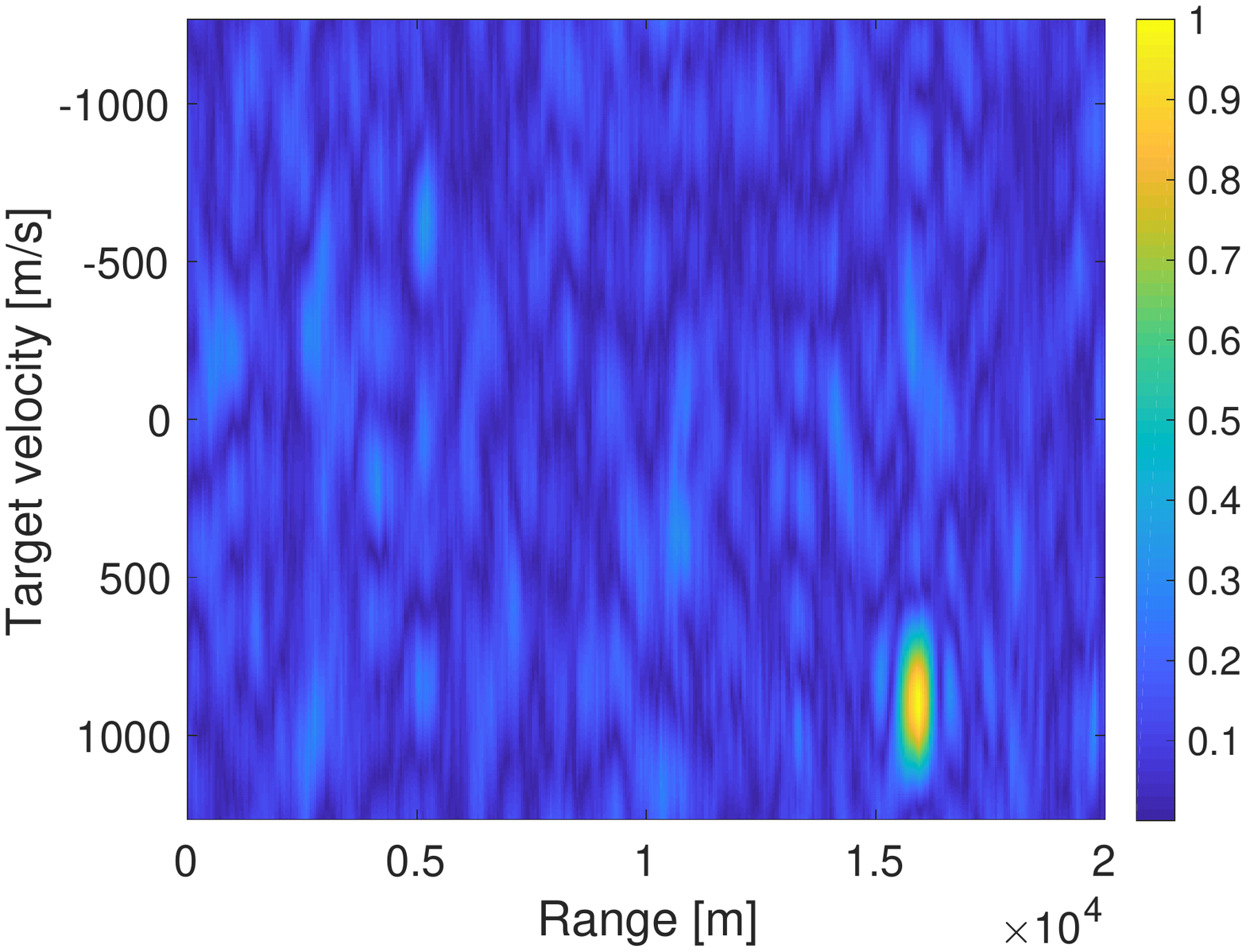} \\
    \scriptsize (b--2) $\text{PRF}_2=1900$ [Hz], $N_2=8$
  \end{tabular}  \hspace{-1.2cm}  \qquad
  \begin{tabular}[b]{c}
    \includegraphics[trim={1cm 6.5cm 4cm 6.5cm}, clip,scale=0.25]{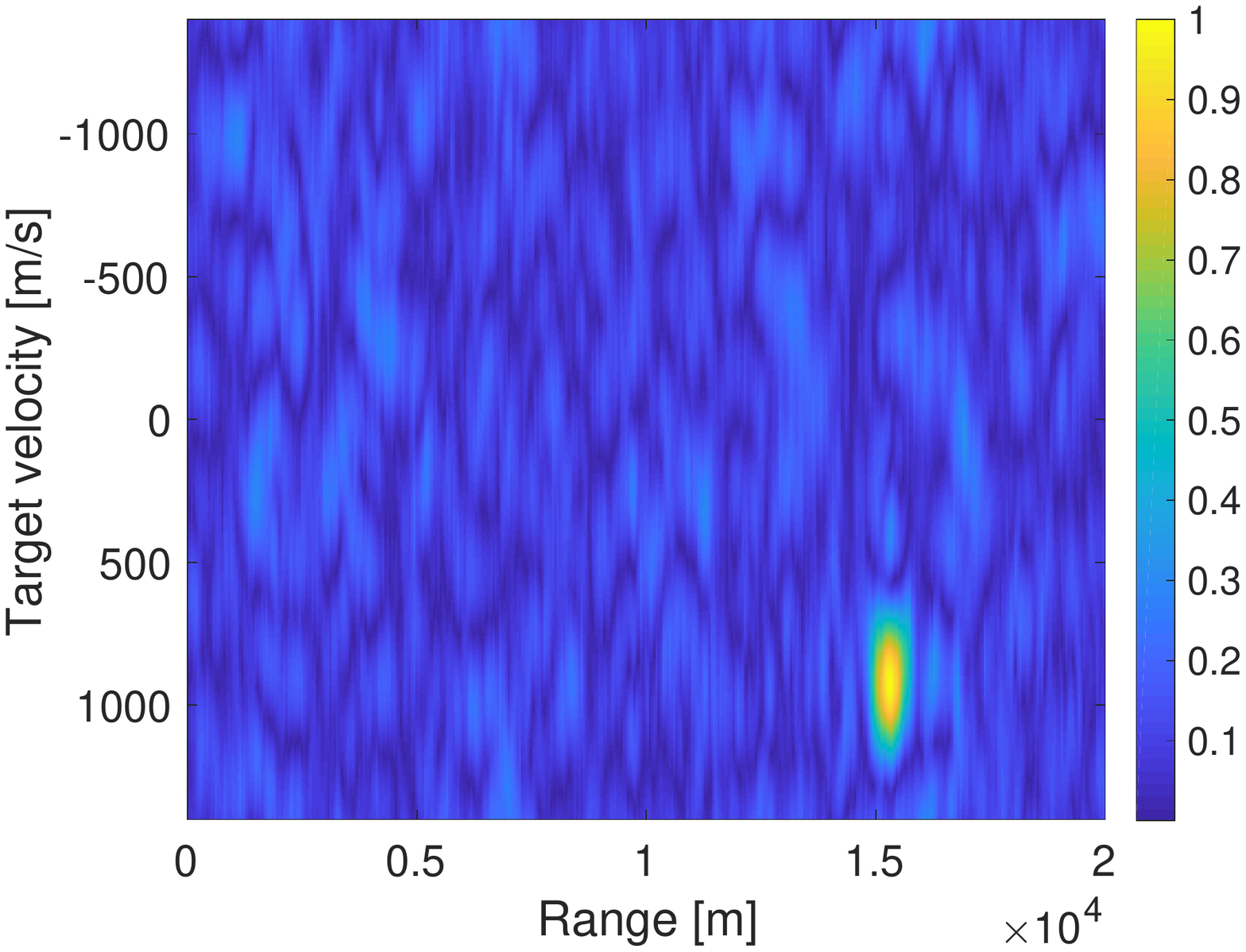} \\
    \scriptsize (b--3) $\text{PRF}_3=2100$ [Hz], $N_3=8$
  \end{tabular} \hspace{-1.2cm} \qquad 
  \begin{tabular}[b]{c}
    \includegraphics[trim={1cm 6.5cm 1.5cm 6.5cm}, clip,scale=0.25]{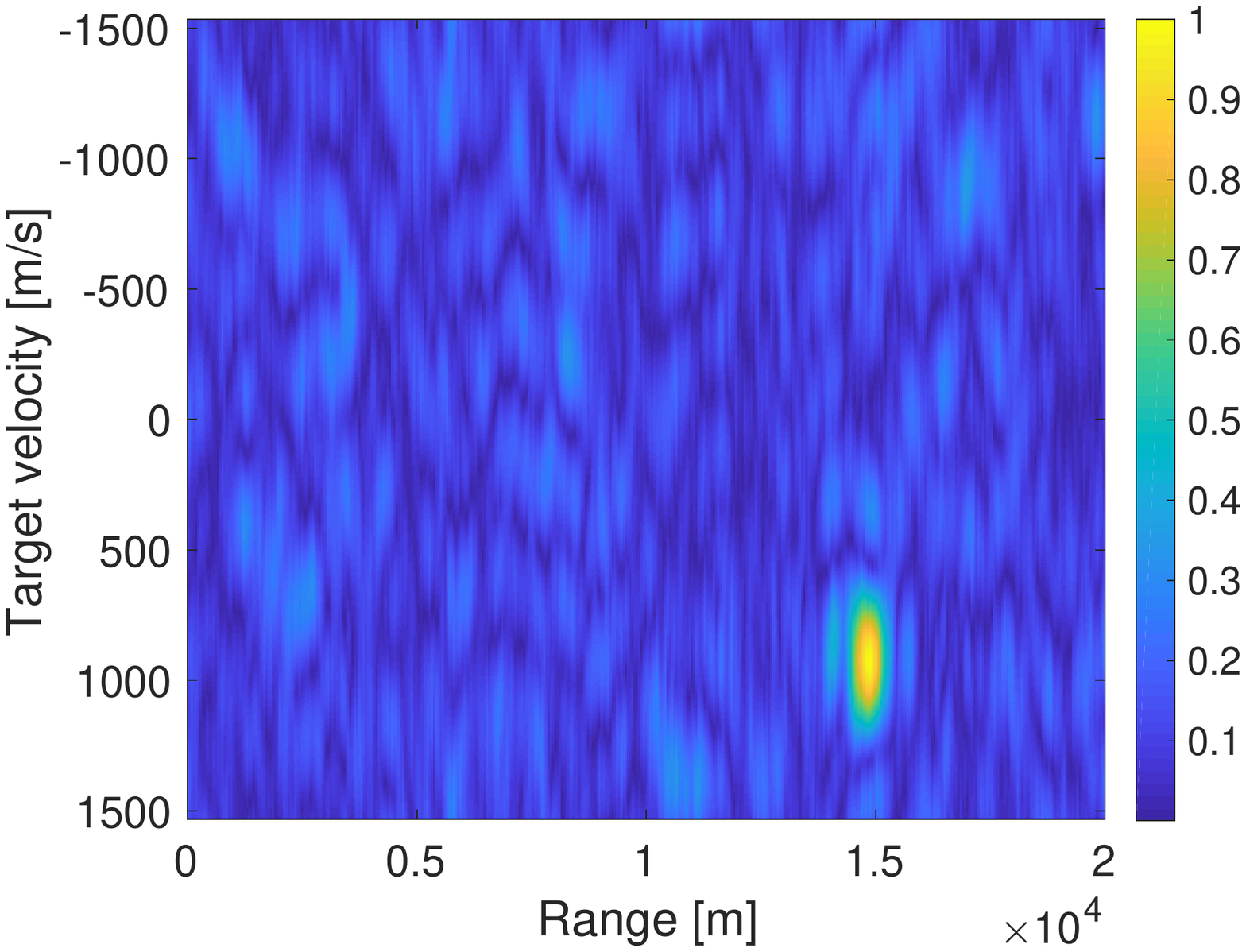} \\
    \scriptsize (b--4) $\text{PRF}_4=2300$ [Hz], $N_4=8$
  \end{tabular}
    %%%%%%%%%%%%%%%%%%%%%%%%%%%%%%%%%%%%%%%%%%%%%%%%
    \begin{tabular}[b]{c}
  \includegraphics[trim={1cm 6.5cm 4cm 6.5cm}, clip,scale=0.25]{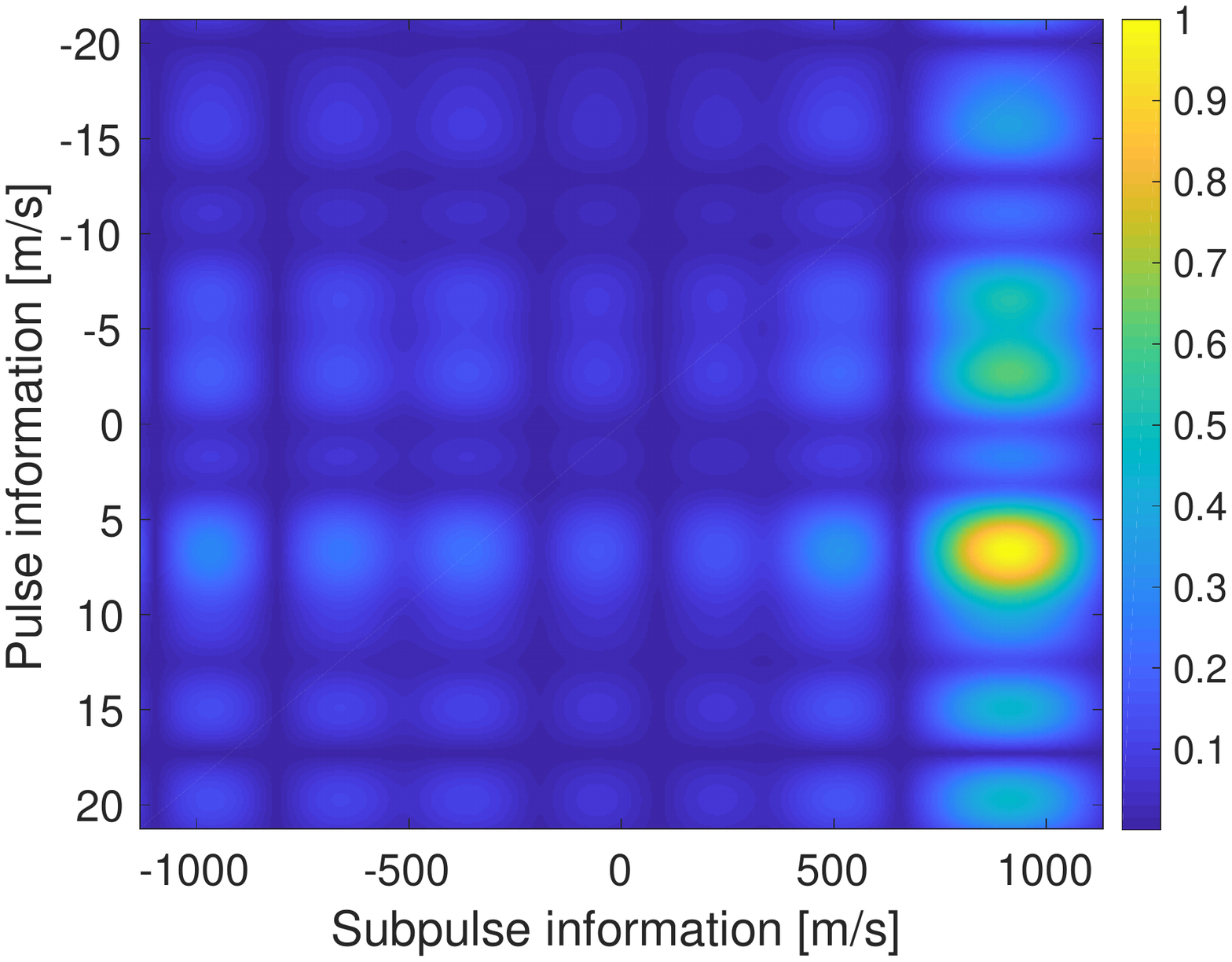} \\
    \scriptsize (c--1) $\text{PRF}_1=1700$ [Hz]
  \end{tabular} \hspace{-1.2cm} \qquad 
  \begin{tabular}[b]{c}
    \includegraphics[trim={1cm 6.5cm 4cm 6.5cm}, clip,scale=0.25]{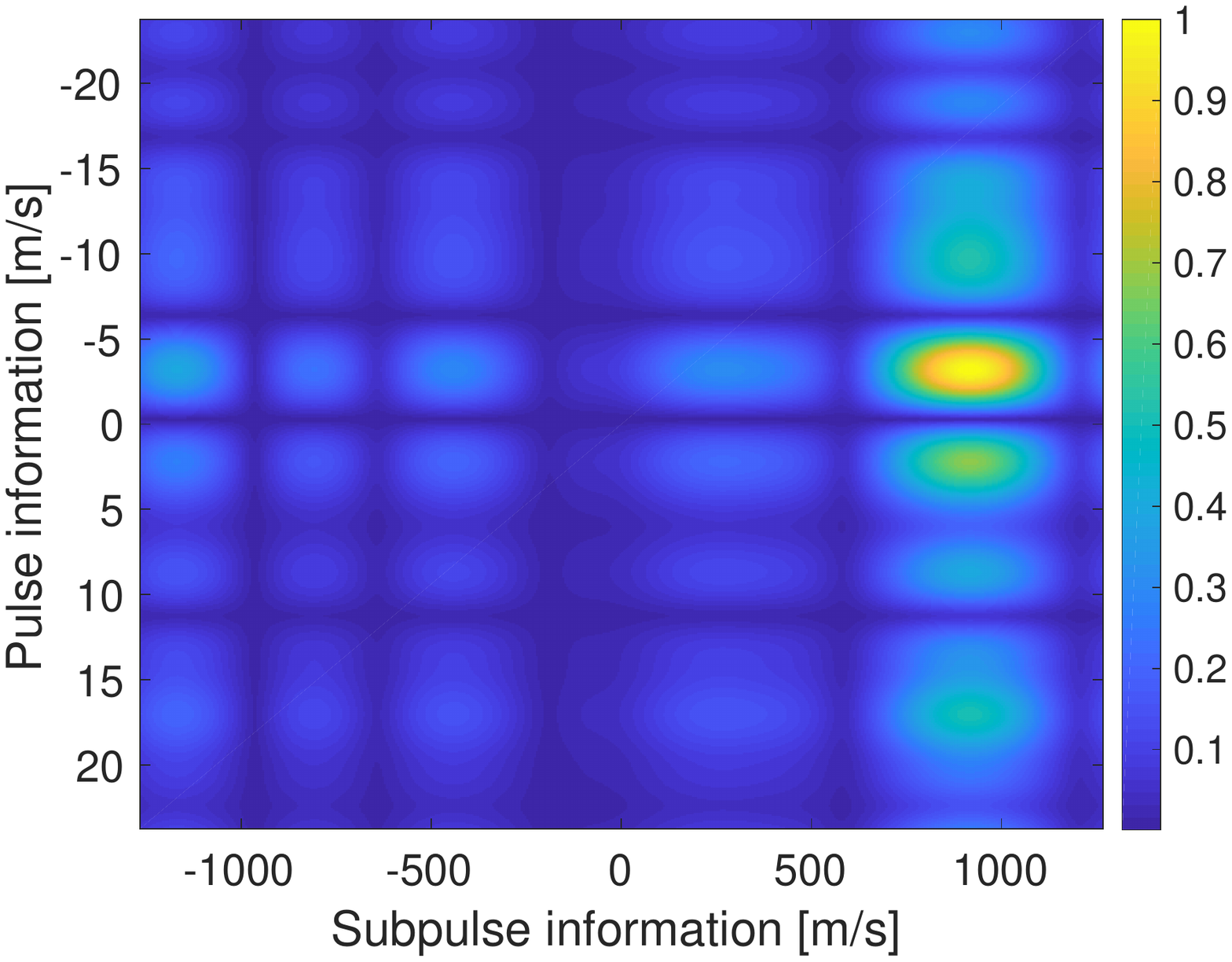} \\
    \scriptsize (c--2) $\text{PRF}_2=1900$ [Hz]
  \end{tabular}  \hspace{-1.2cm}  \qquad
  \begin{tabular}[b]{c}
    \includegraphics[trim={1cm 6.5cm 4cm 6.5cm}, clip,scale=0.25]{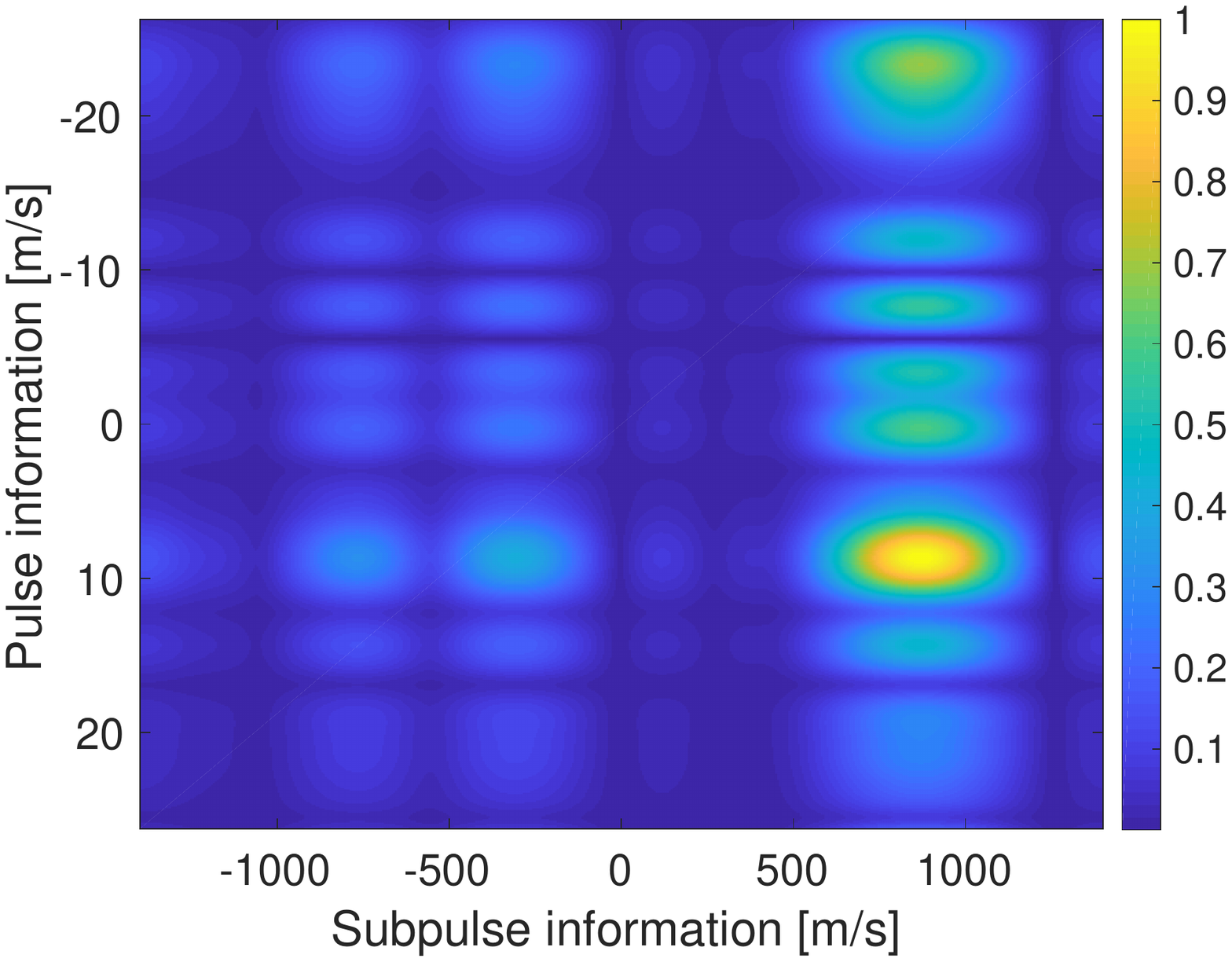} \\
    \scriptsize (c--3) $\text{PRF}_3=2100$ [Hz]
  \end{tabular} \hspace{-1.2cm} \qquad 
  \begin{tabular}[b]{c}
    \includegraphics[trim={1cm 6.5cm 1.5cm 6.5cm}, clip,scale=0.25]{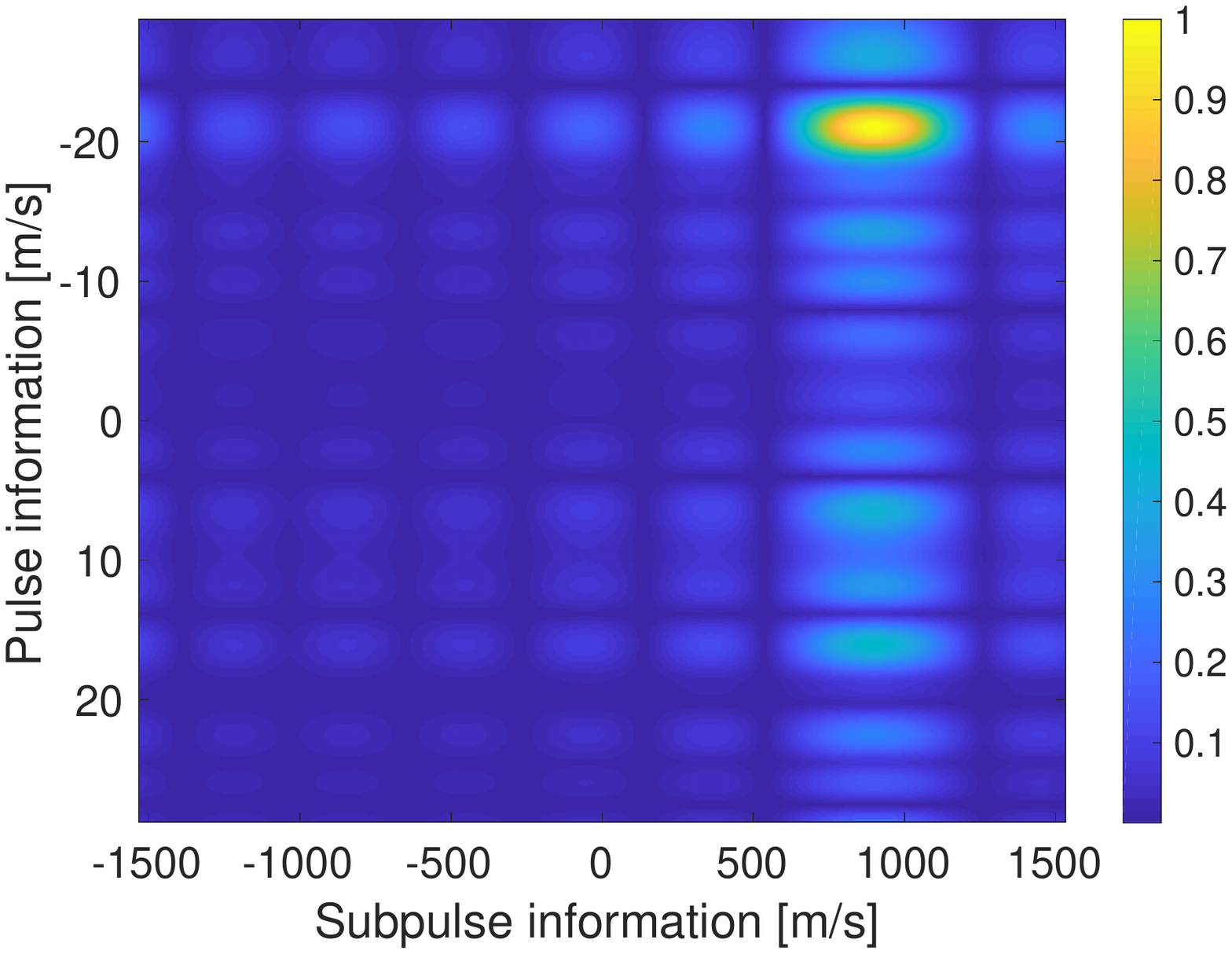} \\
    \scriptsize (c--4) $\text{PRF}_4=2300$ [Hz]
  \end{tabular}
  %%%%%%%%%%%%%%%%%%%%%%%%%%%%%%%%%%%%%%%%%%%%%%%%
  \caption{Doppler estimation.}
  \label{fig: Doppler detecion simulation}
  \hrulefill
\end{figure*}
%=========================================================
In this section, we illustrate through Fig.~\ref{fig: Doppler detecion simulation} how the Doppler estimation process is carried out.
Latter, we validate our derived expressions by means of Monte-Carlo simulations\footnote{The number of realizations in Monte-Carlo simulations was set to $10^6$.}.
To do so, we make use of the following radar setup: $\text{PRF}_1=700$ [Hz], $\text{PRF}_2=1100$ [Hz], $\text{PRF}_3=1300$ [Hz], $\text{PRF}_4=1700$ [Hz], $L=\mathcal{M}=4$, $\mathit{f}_R=6 \ [\text{GHz}]$, $\tau= 25 \  [\mu s] $, $\lambda = 0.05 \ [\text{m}]$, $M_1=11$, $M_2=13$, $M_3=17$, $M_4=19$, and $N_i=8 \ \forall i \in \left\{1,2,3,4\right\}$. 
In addition, we consider a linear frequency-modulated pulse with bandwidth $B=2 \ [\text{MHz}]$.

Fig.~\ref{fig: Doppler detecion simulation} illustrates the output data after the 2D-DFT blocks. In this simulation example, we placed a target at an initial range of 10~[Km], traveling with a constant velocity of $\mathit{v}_t=900$ [m/s] in the opposite direction of the radar (i.e., the target is receding).
Fig.~\ref{fig: Doppler detecion simulation}(a) shows the normalized output data -- Velocity vs Range -- using PP. Observe that in all 4 scenarios, the target at 10 [Km] is unlikely to be detected due to the high loss in SNR.
On the other hand, Fig.~\ref{fig: Doppler detecion simulation}(b) shows the normalized output data -- Velocity vs Range -- using SP. Observe that the loss in SNR is partially mitigated by means of SP. Therefore, the target located at 10~[Km] can now be easily be detect without further processing. 
At last, Fig.~\ref{fig: Doppler detecion simulation}(c) shows the combined pulse and subpulse information.
Note in Fig.~\ref{fig: Doppler detecion simulation}(c) that SP provides a better intuition about the target location, but due to its poor discretization, it is not sufficient to determine the exact velocity. Conversely, PP provides a better discretization but, unfortunately, its velocity estimation is more likely to be ambiguous. 
Thus, by combining SP and the CCRT, we provide the system a high capability to unfold the target's true velocity.

Fig.~\ref{fig:PD1vsrho} shows $\text{PD}_i$ versus $\text{SNR}_1$ using  different values of $M_i$. 
Note how radar performance improves as $M_i$ increases, requiring a lower SNR for a given PD. 
This is because when increasing $M_i$, we are, in fact, increasing the compressed response of PP by means of coherent integration.
In particular, for a fixed $\text{SNR}_1=10$ [dB], we obtain the following probabilities of detection: $\text{PD}_1=0.66$ for $M_1=7$; $\text{PD}_2=0.78$ for $M_2=11$; $\text{PD}_3=0.85$ for $M_3=13$; and $\text{PD}_4=0.93$ for $M_4=17$.
Also, observe that for the high and medium SNR regime, our derived expression matches perfectly the PD of \cite[Eq. (28)]{silva18}. Nevertheless, there is a small difference in the PD for the low SNR regime. This occurs because if the compressed response of PP is less than the background noise, then the intersection probability in \eqref{th: proposition 1} will be less than the probability of \footnotesize$ \bigcap_{k=1}^{M_i-1} \mathcal{A}_{k,i}$ \normalsize. 
For example, given $\text{SNR}_1=4$ [dB] and $M_1=7$, we obtain $\text{PD}_1=0.15$ with our proposed SP--plus--CCRT technique, and $\text{PD}_1=0.18$ with \cite[Eq. (28)]{silva18}.
However, this small reduction in the PD is compensated by a greater reduction in the PFA, as shall be seen next.
%=========================================================
\begin{figure}[t!]
\begin{center}
\includegraphics[trim={0cm 0cm 0cm 0cm}, clip, scale=0.4]{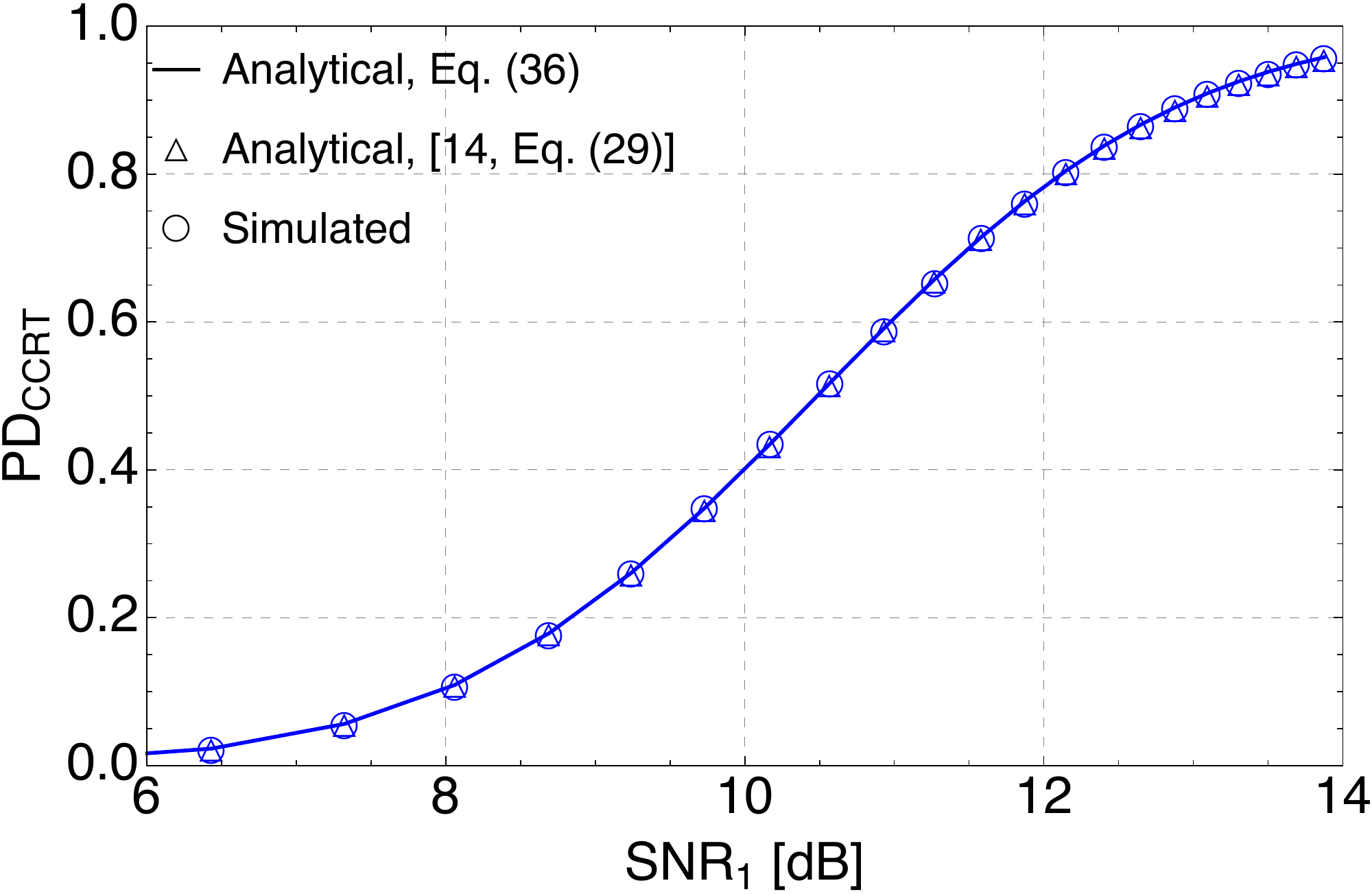}
\caption{$\text{PD}_{\text{CCRT}}$ vs $\text{SNR}_1$ using $N_i=8$, $\lambda_{1,i}=0.5$, $\lambda_{2,i}=0.99$, $\mathcal{M}=4$, and different values of $M_i$ ($i \in \left\{1,2,3,4 \right\}$).
}
\label{fig:PDCRT1} 
\end{center}
\end{figure}
%=========================================================
%=========================================================
\begin{figure}[t!]
\begin{center}
\includegraphics[trim={0cm 0cm 0cm 0cm}, clip, scale=0.4]{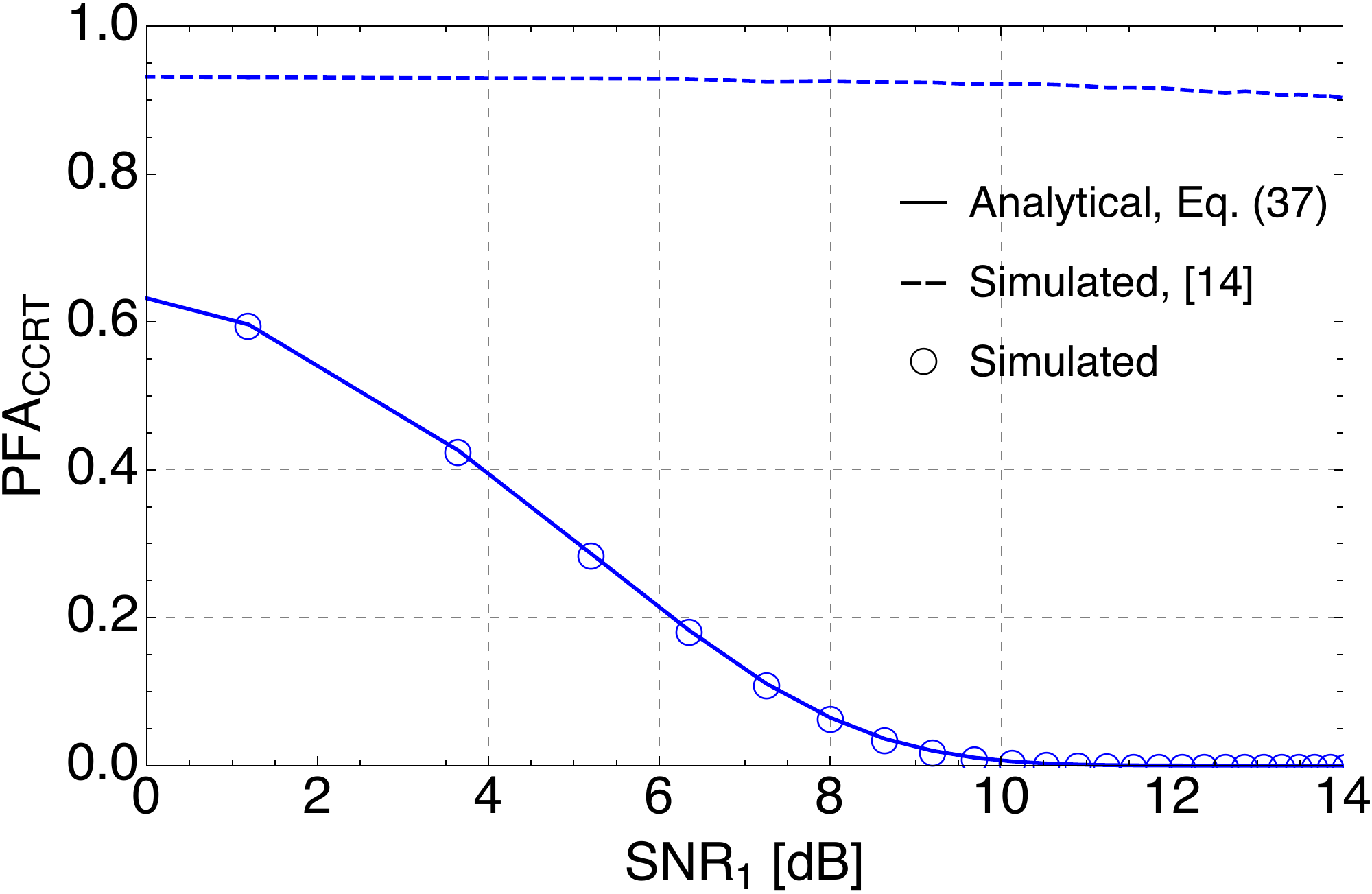}
\caption{$\text{PMD}_{\text{CCRT}}$ vs $\text{SNR}_1$ using $N_i=8$, $\lambda_{1,i}=0.5$, $\lambda_{2,i}=0.99$, $\mathcal{M}=4$, and different values of $M_i$ ($i \in \left\{1,2,3,4 \right\}$).}
\label{fig:PDCRT2} 
\end{center}
\end{figure}
%=========================================================

Fig.~\ref{fig: PD1vsm1} shows $\text{PFA}_i$ versus $\text{SNR}_1$ using different values for $M_i$. 
Observe how $\text{PFA}_i$ decreases as $M_i$ increases. This occurs because as we increase $M_i$, the received target echo becomes stronger compared to the noise background.
For example, for a fixed $\text{SNR}_1=5$ [dB], we obtain the following probabilities of false alarm: $\text{PFA}_1=0.83$ for $M_1=7$; $\text{PFA}_2=0.77$ for $M_2=11$; $\text{PFA}_3=0.73$ for $M_3=13$; and $\text{PFA}_4=0.60$ for $M_4=17$.
More interesting, observe how $\text{PFA}_i$ decays rapidly compared to \cite{silva18}.
This difference in $\text{PFA}_i$ is because intuitively SP acts as a backup detection process. That is, since the compressed response of SP is greater of the PP response (for high-velocity targets), then the probability in \eqref{th: corollary 1} is lower than the probability of \footnotesize$\underset{k=1}{\overset{M_i-1}{\bigcup }} \mathcal{C}_{k,i}$ \normalsize.
For example, using the classic PP technique~\cite{silva18}, we obtain the following probabilities of false alarm: $\text{PFA}_1=0.96$ for $M_1=7$; $\text{PFA}_2=0.97$ for $M_2=11$; $\text{PFA}_3=0.98$ for $M_3=13$; and $\text{PFA}_4=0.99$ for $M_3=17$. 

Finally, Figs. \ref{fig:PDCRT1} and \ref{fig:PDCRT2} show $\text{PD}_{\text{CCRT}}$ and $\text{PFA}_{\text{CCRT}}$ versus $\text{SNR}_1$, respectively. 
Observe in Fig. \ref{fig:PDCRT1}, the perfect agreement between \eqref{eq: PD Poisson binomial distribution CCRT} and~\cite[Eq. (29)]{silva18}.
Hence, in this case, we have no advantage when using SP--plus--CCRT.
On the other hand, observe in Fig.~\ref{fig:PDCRT2}, the high difference in the PFA between of \eqref{eq: PFA Poisson binomial distribution CCRT} and that in~\cite{silva18}. In this case, the use of SP--plus--CCRT improves radar performance by considerably reducing the false alarms. For instance, for given $\text{SNR}_1=2$ [dB], we obtain probabilities of $\text{PFA}_{\text{CCRT}}=0.94$ using PP--plus--CCRT, and $\text{PFA}_{\text{CCRT}}=0.54$ using SP--plus--CCRT.

\section{Conclusion}
\label{sec: Conclusion}
In this work, we provided a thorough statistical analysis on Doppler estimation when both SP and the CCRT were employed.
To do so, we derived novel and closed-form expressions for the PD and PFA.
Moreover, a comparison analysis between our proposed SP--plus--CCRT technique and the classic PP--plus--CCRT was carried out.
Numerical results and Monte-Carlo simulations corroborated the validity of our expressions and showed that the PFA when using SP--plus--CCRT technique was greatly reduced compared to~\cite{silva18}, thereby enhancing radar detection.

\begin{appendices}
\section{Proof of Proposition I}
\label{app: Derivation 1}
Applying \cite[Eq. (5.48)]{leon94} and using the fact that $X_{k,i}$ and $Y_{l,i}$ are independent RVs, \eqref{th: proposition 1} can be rewritten as follows:
\begin{align}
    \label{app:}
     \nonumber  \text{PD}_i = & \int _0^{\infty }\int _0^{\infty } \left( \prod_{k=1}^{M_i-1} \text{Pr} \left[ X_{k,i}< r_{1,i}| R_{1,i}=r_{1,i} \right]  \right)\\
     \nonumber & \times \left( \prod_{l=1}^{N_i-1} \text{Pr} \left[ Y_{l,i}<r_{2,i}| R_{2,i}=r_{2,i} \right]  \right)\\
     & \times \mathit{f}_{R_{1,i},R_{2,i}}(r_{1,i},r_{2,i}) \  \text{d}r_{1,i} \  \text{d}r_{2,i}.
\end{align}
Now, with the aid of \cite[Eq. (4.11)]{leon94} and taking into account that $X_{k,i}$ and $Y_{l,i}$ are identically distributed RVs, yields
\begin{align}
    \label{eq:first PD}
    \nonumber  \text{PD}_i = & \int _0^{\infty }\int _0^{\infty } \left(\int_0^{r_{1,i}} f_{X_{1,i}}(x_{1,i}) \ \text{d} x_{1,i} \right)^{M_i-1} \\
    \nonumber & \times \left(\int_0^{r_{2,i}} f_{Y_{1,i}}(y_{1,i}) \ \text{d} y_{1,i} \right)^{N_i-1} \\
    & \times \mathit{f}_{R_{1,i},R_{2,i}}(r_{1,i},r_{2,i}) \  \text{d}r_{1,i} \  \text{d}r_{2,i}.
\end{align}
Replacing \eqref{eq:PDF x1}--\eqref{eq:PDF x1 x2} in \eqref{eq:first PD}, we obtain
\begin{align}
    \label{eq:PD_expand1}
    \nonumber \text{PD}_i   =&  \int _0^{\infty }\int _0^{\infty } \underbrace{\left(\int _0^{r_{1,i}} \frac{x_{1,i} \exp\left(- \frac{ x_{1,i}^2}{2 \sigma_{1,i}^2} \right)}{\sigma_{1,i}} \text{d}x_{1,i} \right)^{M_i-1}}_{\triangleq \ \mathcal{I}_1} \\
    \nonumber &  \times \underbrace{\left(\int _0^{r_{2,i}} \frac{y_{1,i} \exp\left(- \frac{ y_{1,i}^2}{2 \sigma_{2,i}^2} \right)}{\sigma_{2,i}} \text{d}y_{1,i}  \right)^{N_i-1}}_{\triangleq \ \mathcal{I}_2} \\
    \nonumber & \times \int _0^{\infty } \exp (-\xi_i t) \exp \left(-\textbf{m}_i\right)  I_0\left( 2 \sqrt{\textbf{m}_it} \right)  \\
    \nonumber & \times \prod _{p=1}^2 \frac{r_{p,i}}{\Omega_{p,i}^{2}}  \exp \left(- \frac{r_{p,i}^2}{2 \Omega_{p,i}^2} \right) \\
    & \times I_0 \left( \frac{r_{p,i} \sqrt{t \sigma_{p,i}^2 \lambda_{p,i}^2}}{\Omega_{p,i}^2}\right) \text{d}t \  \text{d}r_{1,i} \  \text{d}r_{2,i}.
\end{align}
In order to solve \eqref{eq:PD_expand1}, we must first evaluate $\mathcal{I}_1$ and $\mathcal{I}_2$. In particular, $\mathcal{I}_1$ can be calculated as follows:
\begin{align}
    \label{eq:I1}
    \nonumber \mathcal{I}_1 & \overset{(a)}{=} \left(  1 -\exp \left( - \frac{ r_{1,i}^2}{2 \sigma_{1,i}^2}\right) \right)^{M_i-1} \\
    & \overset{(b)}{=}\sum _{k=0}^{M_i-1} \left(
    \begin{array}{c}
     M_i-1 \\ k \\ \end{array}
    \right)  \left( - \exp \left( - \frac{ r_{1,i}^2}{2 \sigma_{1,i}^2}\right)\right)^{M_i-1-k},
\end{align}
where in step (a), we have developed the inner integral; and in step (b), we have used the binomial Theorem~\cite{leon94}.

Using a similar approach to that used in~\eqref{eq:I1}, $\mathcal{I}_2$ can be calculated as
\begin{align}
    \label{eq:I2}
    \mathcal{I}_2=  \sum _{l=0}^{N_i-1} \left( \begin{array}{c}
     N_i-1 \\ l \\ \end{array} \right)  \left( - \exp \left( - \frac{ r_{2,i}^2}{2 \sigma_{2,i}^2}\right)\right)^{N_i-1-l}.
\end{align}
Inserting \eqref{eq:I1} and \eqref{eq:I2} in~\eqref{eq:PD_expand1}, followed by changing the order of integration\footnote{The change in the order of integration was performed without loss of generality since \eqref{eq:PDF x1}, \eqref{eq:PDF x2} and \eqref{eq:PDF x1 x2} are non-negative real functions~\cite{friedman80}.} and along with minor manipulations, we obtain~\eqref{eq:PDistep2}, displayed at the top of the next page.
%=========================================================
\begin{figure*}[!t]
%\hrulefill
\begin{flushleft}
% \small
\begin{align}
    \label{eq:PDistep2}
    \nonumber \text{PD}_i = & \sum _{k=0}^{M_i-1} \sum _{l=0}^{N_i-1} 1^{k+l} \left( \begin{array}{c}
     M_i-1 \\ k \\ \end{array}
    \right) \left( \begin{array}{c}
     N_i-1 \\ l \\ \end{array} \right)  \int _0^{\infty }\exp (- \xi_{i} t) \exp \left(-\textbf{m}_i\right)  I_0\left( 2 \sqrt{\textbf{m}_it} \right) \\
     \nonumber & \times \underbrace{\int _0^{\infty } \left( - \exp \left( - \frac{ r_{1,i}^2}{2 \sigma_{1,i}^2}\right)\right)^{M_i-1-k}  \frac{r_{1,i}}{\Omega_{1,i}^{2}}  \exp\left(- \frac{r_{1,i}^2}{2 \Omega_{1,i}^2} \right) I_0 \left( \frac{r_{1,i} \sqrt{t \sigma_{1,i}^2 \lambda_{1,i}^2}}{\Omega_{1,i}^2}\right)  \text{d}r_{1,i} }_{\triangleq \ \mathcal{I}_{3}}   \\
     & \times \underbrace{\int _0^{\infty } \left( - \exp \left( - \frac{ r_{2,i}^2}{2 \sigma_{2,i}^2}\right)\right)^{N_i-1-l} \frac{r_{2,i}}{\Omega_{2,i}^{2}}  \exp\left(- \frac{r_{2,i}^2}{2 \Omega_{2,i}^2} \right) I_0 \left( \frac{r_{2,i} \sqrt{t \sigma_{2,i}^2 \lambda_{2,i}^2}}{\Omega_{2,i}^2}\right)  \text{d}r_{2,i} }_{\triangleq \ \mathcal{I}_{4}}   \text{d}t.
\end{align}
\normalsize
\end{flushleft}
\hrulefill
\end{figure*}
%=========================================================

Now, it remains to find $\mathcal{I}_3$ and $\mathcal{I}_4$. More precisely, $\mathcal{I}_3$ can be computed~as
% \small
\begin{align}
    \label{eq:I3}
    \nonumber \mathcal{I}_3   \overset{(a)}{=}  & \int _0^{\infty } \left( - \exp \left( - \frac{r_{1,i}^2}{2 \sigma_{1,i}^2}\right)\right)^{M_i-1-k}  \frac{r_{1,i}}{\Omega_{1,i}^{2}} \\
    \nonumber & \times \exp\left(- \frac{r_{1,i}^2}{2 \Omega_{1,i}^2} \right) \sum _{q=0}^{\infty } \frac{\left(\frac{r_{1,i} \sqrt{t \lambda_{1,i}^2 \sigma_{1,i}^2}}{2 \Omega_{1,i}^2}\right)^{2 q}}{q! \  \Gamma (q +1)}  \text{d}r_{1,i} \\
    \nonumber \overset{(b)}{=} &  \frac{(-1)^{-k+M_i+1}}{\Omega _{1,i}^2 \left(\frac{-k+M_i-1}{\sigma _{1,i}^2}+\frac{1}{\Omega _{1,i}^2}\right)} \\
    \nonumber & \times \sum _{q=0}^{\infty } \frac{\left(\frac{t \lambda _{1,i}^2 \sigma _{1,i}^4}{2 \Omega _{1,i}^2 \left( \Omega _{1,i}^2 (-k+M_i-1) +\sigma _{1,i}^2\right)}\right)^q}{q!}\\
    \nonumber  \overset{(c)}{=} & \frac{(-1)^{-k+M_i+1}}{\Omega _{1,i}^2 \left(\frac{-k+M_i-1}{\sigma _{1,i}^2}+\frac{1}{\Omega _{1,i}^2}\right)} \\
    & \times \exp \left(\frac{t \lambda _{1,i}^2 \sigma _{1,i}^4}{2 \Omega _{1,i}^2 \left(\Omega _{1,i}^2(-k+M_i-1) +\sigma _{1,i}^2\right)}\right),
\end{align}
\normalsize
where in step (a), we have used the series representation of the modified Bessel function of the first kind and order zero~\cite[Eq. (03.02.02.0001.01)]{Mathematica}; in step (b), we have solved the integral by first changing the order of integration; finally, in step (c), we have used~\cite[Eq. (01.03.06.0002.01)]{Mathematica} and performed some algebraic manipulations.

In like manner as in~\eqref{eq:I3}, $\mathcal{I}_4$ can be computed as 
% \small
\begin{align}
    \label{eq:I4}
    \nonumber \mathcal{I}_4 = & \frac{(-1)^{-l+N_i+1}}{\Omega _{2,i}^2 \left(\frac{-l+N_i-1}{\sigma _{2,i}^2}+\frac{1}{\Omega _{2,i}^2}\right)} \\
    & \times \exp \left(\frac{t \lambda _{2,i}^2 \sigma _{2,i}^4}{2 \Omega _{2,i}^2 \left( \Omega _{2,i}^2(-l+N_i-1) +\sigma _{2,i}^2\right)}\right).
\end{align}
\normalsize
Now, replacing \eqref{eq:I3} and \eqref{eq:I4} in~\eqref{eq:PDistep2}, we obtain
% \small
\begin{align}
    \label{eq:PD_app2}
    \nonumber \text{PD}_i  = & \sum _{k=0}^{M_i-1} \sum _{l=0}^{N_i-1}  \left( \begin{array}{c}
    M_i-1 \\ k \\ \end{array}
    \right) \left( \begin{array}{c}
    N_i-1 \\ l \\ \end{array} \right) \\
    \nonumber & \times  \exp \left(-\textbf{m}_i\right)  \left( \frac{\sigma_{1,i}^2 (-1)^{-k+M_i+1} }{ \Omega_{1,i}^2 (-k+M_i-1)+\sigma_{1,i}^2}   \right)  \\
    \nonumber & \times \left(\frac{\sigma_{2,i}^2 (-1)^{-l+N_i+1}}{ \Omega_{2,i}^2 (-l+N_i-1)+\sigma_{2,i}^2} \right)\\
    \nonumber & \times \int _0^{\infty }\exp (- \xi_i t)  I_0\left( 2 \sqrt{\textbf{m}_it} \right) \\
    \nonumber & \times \exp \left( \frac{t \lambda_{1,i}^2 \sigma_{1,i}^4}{2\Omega_{1,i}^2 \left( \Omega_{1,i}^2 (-k+M_i-1)+\sigma_{1,i}^2\right)}\right) \\
    & \times  \exp \left( \frac{ t \lambda_{2,i}^2\sigma_{2,i}^4}{2\Omega_{2,i}^2\left( \Omega_{2,i}^2  (-l+N_i-1)+\sigma_{2,i}^2\right)}\right) \text{d}t.
\end{align}
\normalsize
Finally, using the following identity~\cite[Eq. (1.11.2.4)]{prudnikov92}
% \small
\begin{equation}
    \label{eq:identity}
    \int_{0}^{\infty} \exp ( t b) I_0 (\sqrt{t} a) \  \text{d}t = -\frac{\exp \left(-\frac{a^2}{4b} \right)}{b},
\end{equation}
\normalsize
and after performing some minor simplifications, we can express \eqref{eq:PD_app2} in closed-form as in \eqref{eq:PD_final}, which completes the proof.

\section{Proof of Corollary I}
\label{app: Derivation 2}
By making use of \cite[Coroll. 6]{leon94}, we can express \eqref{th: corollary 1} as 
\small
\begin{align}
    \label{app: PFA1}
    \nonumber  & \text{PFA}_i =  \sum _{k=1}^{M_i-1} \sum _{l=1}^{N_i-1} \text{Pr} \left[ \mathcal{C}_{k,i} \bigcap \mathcal{D}_{l,i} \right] \\
    \nonumber & \ - \underset{k<p,l<q}{\sum _{k=1}^{M_i-1} \sum _{l=1}^{N_i-1} \sum _{p=2}^{M_i-1} \sum _{q=2}^{N_i-1} }  \text{Pr} \left[ \mathcal{C}_{k,i} \bigcap \mathcal{D}_{l,i} \bigcap \mathcal{C}_{p,i} \bigcap \mathcal{D}_{q,i}  \right] + \hdots \\
    & \ + (-1)^{M_i-N_i-1} \text{Pr} \left[\mathcal{C}_{1,i} \bigcap \mathcal{D}_{1,i} \bigcap \hdots \bigcap \mathcal{C}_{M_i-1,i} \bigcap \mathcal{D}_{N_i-1,i}  \right].
\end{align}
\normalsize
Now, we need to find the event probabilities.
First, let us derive the last event probability of \eqref{app: PFA1}, that is,
\begin{align}
    \label{eq: PR 1}
    \nonumber \text{Pr} & \left[\mathcal{C}_{1,i} \bigcap \mathcal{D}_{1,i} \bigcap \hdots \bigcap \mathcal{C}_{M_i-1,i} \bigcap \mathcal{D}_{N_i-1,i}  \right] \\
    \nonumber \overset{a}{=} & \int _0^{\infty }\int _0^{\infty } \left( \prod_{k=1}^{M_i-1} \text{Pr} \left[ X_{k,i}> r_{1,i}| R_{1,i}=r_{1,i} \right]  \right)\\
    \nonumber & \times \left( \prod_{l=1}^{N_i-1} \text{Pr} \left[ Y_{l,i}>r_{2,i}| R_{2,i}=r_{2,i} \right]  \right)\\
    \nonumber & \times \mathit{f}_{R_{1,i},R_{2,i}}(r_{1,i},r_{2,i}) \  \text{d}r_{1,i} \  \text{d}r_{2,i} \\
    \nonumber \overset{b}{=} & \int _0^{\infty }\int _0^{\infty } \left(\int_{r_{1,i}}^{\infty} f_{X_{1,i}}(x_{1,i}) \ \text{d} x_{1,i} \right)^{M_i-1} \\
    \nonumber & \times \left(\int_{r_{2,i}}^{\infty} f_{Y_{1,i}}(y_{1,i}) \ \text{d} y_{1,i} \right)^{N_i-1}  \\
    & \times \mathit{f}_{R_{1,i},R_{2,i}}(r_{1,i},r_{2,i}) \  \text{d}r_{1,i} \  \text{d}r_{2,i},
\end{align}
where in step (a) we have used \cite[Eq. (5.48)]{leon94}; and in step (b) we have used \cite[Eq. (4.11)]{leon94} along with the fact that $X_{k,i}$ and $Y_{l,i}$ are identically distributed RVs.

Replacing \eqref{eq:PDF x1}--\eqref{eq:PDF x1 x2} in \eqref{eq: PR 1}, yields
\begin{align}
    \label{eq: Pr 2}
    \nonumber \text{Pr} & \left[\mathcal{C}_{1,i} \bigcap \mathcal{D}_{1,i} \bigcap \hdots \bigcap \mathcal{C}_{M_i-1,i} \bigcap \mathcal{D}_{N_i-1,i}  \right]
     \\
   \nonumber  = &\int _0^{\infty }\int _0^{\infty } \underbrace{\left(\int _{r_{1,i}}^{\infty} \frac{x_{1,i} \exp\left(- \frac{ x_{1,i}^2}{2 \sigma_{1,i}^2} \right)}{\sigma_{1,i}} \text{d}x_{1,i} \right)^{M_i-1}}_{\triangleq \ \mathcal{I}_5} \\
    \nonumber &  \times \underbrace{\left(\int _{r_{2,i}}^{\infty} \frac{y_{1,i} \exp\left(- \frac{ y_{1,i}^2}{2 \sigma_{2,i}^2} \right)}{\sigma_{2,i}} \text{d}y_{1,i}  \right)^{N_i-1}}_{\triangleq \ \mathcal{I}_6} \\
    \nonumber & \times \int _0^{\infty } \exp (-\xi_i t) \exp \left(-\textbf{m}_i\right)  I_0\left( 2 \sqrt{\textbf{m}_it} \right)  \\
    \nonumber & \times \prod _{p=1}^2 \frac{r_{p,i}}{\Omega_{p,i}^{2}}  \exp \left(- \frac{r_{p,i}^2}{2 \Omega_{p,i}^2} \right) \\
    & \times I_0 \left( \frac{r_{p,i} \sqrt{t \sigma_{p,i}^2 \lambda_{p,i}^2}}{\Omega_{p,i}^2}\right) \text{d}t \  \text{d}r_{1,i} \  \text{d}r_{2,i}.
\end{align}
After some mathematical manipulations, $\mathcal{I}_5$ and $\mathcal{I}_6$ can be calculated, respectively, as
\begin{align}
    \label{eq:I5}
    \mathcal{I}_5  = &\exp \left( - \frac{ r_{1,i}^2 (M_i-1)}{2 \sigma_{1,i}^2}\right)  \\
    \label{eq:I6}
    \mathcal{I}_6= & \exp \left( - \frac{ r_{2,i}^2 (N_i-1)}{2 \sigma_{2,i}^2}\right).
\end{align}
%=========================================================
\begin{figure*}[!t]
%\hrulefill
\begin{flushleft}
\small
\begin{align}
    \label{app: PFA2}
    \nonumber \text{PFA}_i = & \binom{M_i-1}{1} \binom{N_i-1}{1} \frac{\mathcal{Q}_i \left(1,1 \right)}{ \mathcal{P}_i \left(1,1 \right) } \exp \left( -\textbf{m}_i+ \frac{\textbf{m}_i}{\mathcal{P}_i \left(1,1 \right)} \right) - \binom{M_i-1}{2} \binom{N_i-1}{2} \frac{\mathcal{Q}_i \left(2,2 \right)}{ \mathcal{P}_i \left(2,2 \right) } \exp \left( -\textbf{m}_i+ \frac{\textbf{m}_i}{\mathcal{P}_i \left(2,2 \right)} \right)  + \hdots \\
    & + (-1)^{M_i-N_i-1} \binom{M_i-1}{M_i-1} \binom{N_i-1}{N_i-1} \frac{\mathcal{Q}_i \left(M_i-1,N_i-1 \right)}{ \mathcal{P}_i \left(M_i-1,N_i-1  \right) } \exp \left( -\textbf{m}_i+ \frac{\textbf{m}_i}{\mathcal{P}_i \left(M_i-1,N_i-1  \right)} \right) 
\end{align}
\normalsize
\end{flushleft}
\hrulefill
\end{figure*}
%=========================================================
Now, replacing \eqref{eq:I5} and \eqref{eq:I6} in \eqref{eq: Pr 2}, and after solving remaining three integrals by applying the same procedure as in \eqref{eq:PD_app2}, we obtain
\small
\begin{align}
    \label{app: Pr PQ}
    \nonumber \text{Pr} & \left[\mathcal{C}_{1,i} \bigcap \mathcal{D}_{1,i} \bigcap \hdots \bigcap \mathcal{C}_{M_i-1,i} \bigcap \mathcal{D}_{N_i-1,i}  \right] \\
    & =
    \frac{\mathcal{Q}_i \left(M_i-1,N_i-1 \right)}{ \mathcal{P}_i \left(M_i-1,N_i-1  \right) } \exp \left( -\textbf{m}_i+ \frac{\textbf{m}_i}{\mathcal{P}_i \left(M_i-1,N_i-1  \right)} \right),
\end{align}
\normalsize
where $\mathcal{P}_i \left(k,l \right)$ and $\mathcal{Q}_i \left(k,l \right)$ are auxiliary functions defined in \eqref{eq: PQ functions}, and the parameters $k \in \left\{ 1,2,\hdots,M_i-1 \right\}$ and $l \in \left\{ 1,2,\hdots,N_i-1 \right\}$ denote the number of events for $\mathcal{C}_{k,i}$ and $\mathcal{D}_{l,i}$, respectively.
Thus, the remaining event probabilities in \eqref{app: PFA1} can be easily obtained by a proper choice of the parameters $k$ and $l$. For example, for $k=1$ and $l=3$, we obtain
\begin{align}
    \label{}
    \nonumber \text{Pr} &\left[ \mathcal{C}_{1,i} \bigcap \mathcal{D}_{1,i} \bigcap \mathcal{D}_{2,i} \bigcap \mathcal{D}_{3,i}  \right] \\
    &= \frac{\mathcal{Q}_i \left(1,3 \right)}{ \mathcal{P}_i \left(1,3 \right) } \exp \left( -\textbf{m}_i+ \frac{\textbf{m}_i}{\mathcal{P}_i \left(1,3 \right)} \right),
\end{align}
whereas for $k=3$ and $l=2$, we have
\begin{align}
    \label{}
    \nonumber \text{Pr} &\left[ \mathcal{C}_{1,i} \bigcap \mathcal{D}_{1,i} \bigcap \mathcal{C}_{2,i} \bigcap \mathcal{D}_{2,i} \bigcap \mathcal{C}_{3,i}  \right]\\
    & = \frac{\mathcal{Q}_i \left(3,2 \right)}{ \mathcal{P}_i \left(3,2 \right) } \exp \left( -\textbf{m}_i+ \frac{\textbf{m}_i}{\mathcal{P}_i \left(3,2 \right)} \right).
\end{align}
Later, with the aid of \eqref{app: Pr PQ} and after some algebraic manipulations, we can rewrite \eqref{app: PFA1} as in \eqref{app: PFA2}, displayed at the top of the next page. 
Finally, and after minor simplifications, \eqref{app: PFA2} reduces to \eqref{eq: PFA final}, which completes the proof.

\end{appendices}

\bibliographystyle{IEEEtran}
\bibliography{Subpulse}

\begin{thebibliography}{10}
\providecommand{\url}[1]{#1}
\csname url@samestyle\endcsname
\providecommand{\newblock}{\relax}
\providecommand{\bibinfo}[2]{#2}
\providecommand{\BIBentrySTDinterwordspacing}{\spaceskip=0pt\relax}
\providecommand{\BIBentryALTinterwordstretchfactor}{4}
\providecommand{\BIBentryALTinterwordspacing}{\spaceskip=\fontdimen2\font plus
\BIBentryALTinterwordstretchfactor\fontdimen3\font minus
  \fontdimen4\font\relax}
\providecommand{\BIBforeignlanguage}[2]{{%
\expandafter\ifx\csname l@#1\endcsname\relax
\typeout{** WARNING: IEEEtran.bst: No hyphenation pattern has been}%
\typeout{** loaded for the language `#1'. Using the pattern for}%
\typeout{** the default language instead.}%
\else
\language=\csname l@#1\endcsname
\fi
#2}}
\providecommand{\BIBdecl}{\relax}
\BIBdecl

\bibitem{Morris96}
G.~Morris and L.~Harkness, \emph{Airborne Pulsed Doppler Radar}, 2nd~ed.\hskip
  1em plus 0.5em minus 0.4em\relax Norwood, MA, USA: Artech House, 1996.

\bibitem{richards10}
M.~A. Richards, J.~Scheer, W.~A. Holm, and W.~L. Melvin, \emph{Principles of
  Modern Radar: Basic Principles}, 1st~ed.\hskip 1em plus 0.5em minus
  0.4em\relax West Perth, WA, Australia: SciTech, 2010.

\bibitem{trunk78}
G.~V. Trunk, ``Range resolution of targets using automatic detectors,''
  \emph{IEEE Trans. Aerosp. Electron. Syst.}, vol. AES-14, no.~5, pp. 750--755,
  Sept. 1978.

\bibitem{hovanession76}
S.~A. {Hovanessian}, ``An algorithm for calculation of range in a multiple
  {PRF} radar,'' \emph{IEEE Trans. Aerosp. Electron. Syst.}, vol. AES-12,
  no.~2, pp. 287--290, Mar. 1976.

\bibitem{Xia07}
X.-G. Xia and G.~Wang, ``Phase unwrapping and a robust chinese remainder
  theorem,'' \emph{IEEE Signal Process. Lett.}, vol.~14, no.~4, pp. 247--250,
  Apr. 2007.

\bibitem{XLi09}
X.~Li, H.~Liang, and X.~Xia, ``A robust chinese remainder theorem with its
  applications in frequency estimation from undersampled waveforms,''
  \emph{IEEE Trans. Signal Process.}, vol.~57, no.~11, pp. 4314--4322, Nov.
  2009.

\bibitem{Wang10}
W.~Wang and X.~Xia, ``A closed-form robust chinese remainder theorem and its
  performance analysis,'' \emph{IEEE Trans. Signal Process.}, vol.~58, no.~11,
  pp. 5655--5666, Nov. 2010.

\bibitem{Trunk94}
G.~V. Trunk and W.~M. Kim, ``Ambiguity resolution of multiple targets using
  pulse-{Doppler} waveforms,'' \emph{IEEE Trans. Aerosp. Electron. Syst.},
  vol.~30, no.~4, pp. 1130--1137, Oct. 1994.

\bibitem{Garcia19USA}
F.~D.~A. {García}, A.~S. {Guerreiro}, G.~R.~L. {Tejerina}, J.~C.~S. {Santos
  Filho}, G.~{Fraidenraich}, M.~D. {Yacoub}, M.~A.~M. {Miranda}, and
  H.~{Cioqueta}, ``Probability of detection for unambiguous doppler frequencies
  in pulsed radars using the chinese remainder theorem and subpulse
  processing,'' in \emph{Proc. 53rd Asilomar Conference on Signals, Systems,
  and Computers}, Pacific Grove, CA, USA, Nov. 2019, pp. 138--142.

\bibitem{skolnik01}
M.~I. Skolnik, \emph{Introduction to Radar Systems}, 3rd~ed.\hskip 1em plus
  0.5em minus 0.4em\relax Ney York, NY, USA: McGraw-Hill, 2001.

\bibitem{beltrao17}
G.~Beltrao, L.~Pralon, M.~Menezes, P.~Vyplavin, B.~Pompeo, and M.~Pralon,
  ``Subpulse processing for long range surveillance noise radars,'' in
  \emph{Proc. International Conference on Radar Systems (Radar 2017)}, Belfast,
  UK, Oct. 2017, pp. 1--4.

\bibitem{Barreto19}
A.~{Barreto}, L.~{Pralon}, B.~{Pompeo}, G.~{Beltrao}, and M.~{Pralon}, ``{FPGA}
  design and implementation of a real-time subpulse processing architecture for
  noise radars,'' in \emph{Proc. 2019 International Radar Conference (RADAR)},
  Toulon, France, Sept. 2019, pp. 1--6.

\bibitem{Doviak93}
D.~S. Doviak and R.~J. Zrnic, \emph{Doppler Radar and Weather Observations},
  2nd~ed.\hskip 1em plus 0.5em minus 0.4em\relax San Diego, CA, USA: Academic
  Press, 2001.

\bibitem{silva18}
B.~Silva and G.~Fraidenraich, ``Performance analysis of the classic and robust
  chinese remainder theorems in pulsed doppler radars,'' \emph{IEEE Trans.
  Signal Process.}, vol.~66, no.~18, pp. 4898--4903, Sept. 2018.

\bibitem{richards14}
M.~A. Richards, \emph{Fundamentals of Radar Signal Processing}, 2nd~ed.\hskip
  1em plus 0.5em minus 0.4em\relax Ney York, NY, USA: McGraw-Hill, 2014.

\bibitem{barton13}
D.~K. Barton, \emph{Radar Equations for Modern Radar}, 1st~ed.\hskip 1em plus
  0.5em minus 0.4em\relax Massachusetts, MA, USA: Artech House, 2013.

\bibitem{trunk93}
G.~{Trunk} and S.~{Brockett}, ``Range and velocity ambiguity resolution,'' in
  \emph{Proc. Record IEEE Nat. Radar Conf.}, Lynnfield, MA, USA, Apr. 1993, pp.
  146--149.

\bibitem{ferrari97}
A.~Ferrari, C.~Berenguer, and G.~Alengrin, ``Doppler ambiguity resolution using
  multiple {PRF},'' \emph{IEEE Trans. Aerosp. Electron. Syst.}, vol.~33, no.~3,
  pp. 738--751, Jul. 1997.

\bibitem{papoulis02}
A.~Papoulis, \emph{Probability, Random Variables, and Stochastic Processes},
  4th~ed.\hskip 1em plus 0.5em minus 0.4em\relax Ney York, NY, USA:
  McGraw-Hill, 2002.

\bibitem{beaulieu11}
N.~C. Beaulieu and K.~T. Hemachandra, ``Novel representations for the bivariate
  rician distribution,'' \emph{IEEE Trans. Commun.}, vol.~59, no.~11, pp.
  2951--2954, Nov. 2011.

\bibitem{Behnad12}
A.~Behnad, N.~C. Beaulieu, and K.~T. Hemachandra, ``Correction to ``{N}ovel
  representations for the bivariate rician distribution'','' \emph{IEEE Trans.
  Commun.}, vol.~60, no.~6, pp. 1486--1486, Jun. 2012.

\bibitem{abramowitz72}
M.~Abramowitz and I.~A. Stegun, \emph{Handbook of Mathematical Functions with
  Formulas, Graphs, and Mathematical Tables}.\hskip 1em plus 0.5em minus
  0.4em\relax Washington, DC: US Dept. of Commerce: National Bureau of
  Standards, 1972.

\bibitem{Wang93}
Y.~H. Wang, ``On the number of successes in independent trials,''
  \emph{Statistica Sinica}, vol.~3, no.~2, pp. 295--312, 1993.

\bibitem{leon94}
A.~Leon-Garcia, \emph{Probability and Random Processes for Electrical
  Engineering}, 3rd~ed.\hskip 1em plus 0.5em minus 0.4em\relax New Jersey, NJ,
  USA: Pearson Prentice Hall, 1994.

\bibitem{friedman80}
H.~Friedman, ``A consistent {F}ubini-{T}onelli theorem for nonmeasurable
  functions,'' \emph{Illinois J. Math.}, vol.~24, no.~3, pp. 390--395, 1980.

\bibitem{Mathematica}
\BIBentryALTinterwordspacing
{Wolfram Research, Inc. (2018)}, \emph{Wolfram Research}, Accessed: Sept. 19,
  2018. [Online]. Available: \url{http://functions.wolfram.com}
\BIBentrySTDinterwordspacing

\bibitem{prudnikov92}
A.~P. Prudnikov, Y.~A. Bry{\v c}kov, and O.~I. Mari{\v c}ev, \emph{Integral and
  Series: {V}ol. 2}, 2nd~ed., Fizmatlit, Ed.\hskip 1em plus 0.5em minus
  0.4em\relax Moscow, Russia: Fizmatlit, 1992.

\end{thebibliography}

\end{document}